\documentclass[11pt]{article}
\usepackage{amsmath}
\usepackage[sc]{mathpazo}

\usepackage[colorlinks]{hyperref}
\hypersetup{linkcolor=[rgb]{0,0.4,0.7},filecolor=[rgb]{0,0.4,0.7},citecolor=[rgb]{0,0.4,0.7},urlcolor=[rgb]{0,0,0.7}}
\usepackage{xspace}
\usepackage{fullpage}
\usepackage{boxedminipage}
\usepackage[boxed]{algorithm}
\usepackage{epigraph}
\usepackage{framed}
\usepackage[framemethod=tikz]{mdframed}
\usepackage{titlesec}
\usepackage{lipsum}%

\usepackage{tikz}
\usetikzlibrary{shapes.geometric}
\usetikzlibrary{arrows}
\usetikzlibrary{arrows.meta}
\usetikzlibrary{patterns}
\usetikzlibrary{shapes.misc}

\newcommand{\congest}{${\mathsf{CONGEST}}$}

\newcommand{\paren}[1]{\left (#1\right)}

\newcommand{\curparen}[1]{\left \{#1\right\}}
\newcommand{\curparenn}[1]{\{#1\}}

\newcommand{\dilation}{\mbox{\tt d}}
\newcommand{\congestion}{\mbox{\tt c}}

\newcommand{\Diam}{\mathsf{Diam}}

\newcommand{\PSMEnc}{\mathsf{PSM.Enc}}
\newcommand{\PSMDec}{\mathsf{PSM.Dec}}
\newcommand{\View}{\mathsf{View}}
\newcommand{\Sim}{\mathsf{Sim}}
\newcommand{\PSM}{\mathsf{PSM}}
\newcommand{\polylog}{\mathsf{polylog}}

\def\cA{{\cal A}}
\def\cC{{\cal C}}
\def\cN{{\cal N}}

\renewcommand{\paragraph}[1]{\vspace{0.15cm}\noindent {\bf #1}}
\titlespacing*{\section}{0pt}{1.1\baselineskip}{\baselineskip}
\titlespacing*{\subsection}{-5pt}{\baselineskip}{\baselineskip}

\usepackage{amsthm,amsmath,amssymb}
\usepackage{cleveref,aliascnt}

\newtheorem{theorem}{Theorem}

\newtheorem{definition}{Definition}

\newtheorem{lemma}{Lemma}

\newtheorem{claim}{Claim}
\newtheorem{remark}{Remark}
\newtheorem{corollary}{Corollary}

\usepackage{tikz}
\usepackage{relsize}
\usepackage{ctable}

\crefname{claim}{Claim}{Claims}
\crefname{observation}{Observation}{Observations}

\newcommand{\A}{{\mathcal A}}

\newcommand{\bit}{\{0,1\}}

\newcommand{\ie}  {i.e.,\ }
\newcommand{\eg}  {e.g.,\ }

\newcommand{\etalcite}[1]{{et~al.~\cite{#1}}}
\newcommand{\poly}{\mathsf{poly}}

\newcommand{\ith}[1]{{#1}\textsuperscript{th}}

\newcommand{\ignore}[1]{}

\title{Distributed Algorithms Made Secure:\\A Graph Theoretic Approach}
\author{Merav Parter \thanks{Department of Computer Science and
		Applied Mathematics, Weizmann Institute of Science, Israel.}
	\and Eylon Yogev\footnotemark[2]
}

\date{}

\begin{document}
\maketitle
\begin{abstract}
In the area of distributed graph algorithms a number of network's entities with 
local views solve some computational task by exchanging messages with their 
neighbors. Quite unfortunately, an inherent property of most existing 
distributed algorithms is that throughout the course of their execution, the 
nodes get to learn not only their own output but rather learn quite a lot on 
the inputs or outputs of many other entities. This leakage of information might 
be a major 
obstacle in settings where the output (or input) of network's individual is a 
private information (\eg distributed networks of selfish agents, 
decentralized digital currency such as Bitcoin).

While being quite an unfamiliar notion in the classical distributed setting,  
the notion of secure multi-party computation (MPC) is one of the main 
themes in the Cryptographic community. The existing secure MPC protocols do not 
quite fit the framework of classical distributed models in which only messages 
of bounded size are sent on graph edges in each round.
In this paper, we introduce a new framework for \emph{secure distributed graph 
algorithms} and provide the first \emph{general compiler} that takes any 
``natural'' non-secure distributed algorithm that runs in $r$ rounds, and turns 
it into a secure algorithm that runs in $\widetilde{O}(r \cdot D \cdot 
\poly(\Delta))$ rounds where $\Delta$ 
is the maximum degree in the graph and $D$ is its diameter. 
A ``natural'' 
distributed algorithm is one where the local computation at each node can be performed in 
polynomial time. An interesting advantage of our approach is that it allows one to decouple between the price of locality and the price of \emph{security} of a given graph function $f$. The security of the compiled algorithm is 
information-theoretic but holds only against a semi-honest adversary that 
controls a single node in the network. 

This compiler is made possible due to a new combinatorial structure called 
\emph{private neighborhood trees}: a collection of $n$ trees $T(u_1),\ldots, 
T(u_n)$, one for each vertex $u_i \in V(G)$, such that each tree $T(u_i)$ spans 
the neighbors of $u_i$ {\em without going through $u_i$}. Intuitively, each 
tree $T(u_i)$ allows all neighbors of $u_i$ to 
exchange a \emph{secret} that is hidden from $u_i$, which is the basic 
graph infrastructure of the compiler. In a 
$(\dilation,\congestion)$-private neighborhood trees each tree $T(u_i)$ has 
depth at most $\dilation$ and each edge $e \in G$ appears in at most 
$\congestion$ different trees. We show a 
construction of private neighborhood trees with 
$\dilation=\widetilde{O}(\Delta \cdot D)$ and $\congestion=\widetilde{O}(D)$, 
both these bounds are \emph{existentially} optimal.
\end{abstract}

\thispagestyle{empty}
\newpage
\tableofcontents
\thispagestyle{empty}
\newpage

\setcounter{page}{1}
\section{Introduction}
In \emph{distributed graph algorithms} (or network algorithms) a number of 
individual entities are connected via a potentially large network.  
Starting with the breakthrough by Awerbuch \etalcite{AwerbuchGLP89}, 
and the 
seminal work of Linial \cite{linial1992locality}, Peleg \cite{Peleg:2000} and 
Naor and Stockmeyer \cite{naor1995can}, the area 
of distributed graph algorithms is growing rapidly. Recently, it has been receiving considerably more 
theoretical and practical attention motivated by the spread of 
multi-core computers, cloud computing, and distributed 
databases. We consider the standard synchronous message passing model (the \congest\ 
model) where in 
each round $O(\log n)$ bits can be transmitted over every edge where $n$ is the 
number of entities.

The common principle underlying all distributed graph algorithms (regardless of 
the model specification) is that the input of the algorithm is given in a
\emph{distributed format}. Consequently the goal of each vertex is to 
compute its \emph{own} part of the output, \eg whether it is a member of a 
computed 
maximal independent set, its own color in a valid coloring of the graph, its 
incident edges in the minimum spanning tree, or its 
chosen edge for a maximal matching solution.
In most distributed algorithms, throughout execution,  
vertices learn much more than merely their own output but 
rather collect additional information on the input or output of (potentially) 
many other vertices in the network. This seems inherent in many distributed 
algorithms, as the output of one 
node is used in the computation of another.
For instance, most randomized coloring (or 
MIS) algorithms 
\cite{luby1986simple,barenboim2013distributed,
barenboim2016locality,harris2016distributed,
ghaffari2016improved,chung2017distributed} are based on the vertices exchanging 
their current color with their 
neighbors in order to decide whether they are legally colored. 

In cases where the data is sensitive or private, these algorithms may raise 
security concerns. To exemplify this point, consider the task of computing the 
average salary in a distributed network. This is a rather
simple distributed task: construct a BFS tree and let the nodes send their 
salary from the leaves to the root where each intermediate node 
sends to its parent in the tree, the sum of all salaries received from its 
children. While the output goal has been 
achieved, privacy has been compromised as intermediate nodes learn more 
information regarding the salaries of their subtrees. 
Additional motivation for secure distributed computation includes private 
medical data, networks of selfish agents with private utility functions, and 
decentralized digital currencies such as the Bitcoin.

The community of distributed graph algorithms is commonly concerned with two 
primary challenges, namely, locality (i.e., communication is only performed 
between neighboring nodes) and congestion (i.e., communication links have 
bounded bandwidth). Security is usually not specified as a desired requirement 
of the distributed algorithm and the main efficiency criterion is the round 
complexity (while respecting bandwidth limitation).

Albeit being a rather virgin objective in the area of distributed graph algorithms, 
the notion of security in multi-party computation (MPC) is one of the 
main themes in the Cryptographic community.
Broadly speaking, \emph{secure} MPC protocols allow 
parties to jointly compute a function $f$ of their inputs without revealing 
anything about their inputs except the output of the function. 
There has been tremendous 
work on MPC protocols, starting from general feasibility results 
\cite{Yao82b,GoldreichMW87,BenorGW88,ChaumCD88} that apply to any functionality 
to protocols that are designed to be 
extremely efficient for 
specific functionalities \cite{Ben-DavidNP08,Ben-EfraimLO16}. 
There is also a wide range of security notions:
information-theoretic security or security that is based on computational assumptions, the adversary is 
either semi-honest or malicious\footnote{A semi-honest adversary
does not deviate from the described protocol, but may run any computation 
on the received transcript to gain additional information. A malicious 
adversary might arbitrarily deviate 
from 
the protocol.} and in might collude with several 
parties. 

Most MPC protocols are designed for the clique networks where 
every two parties have a secure channel between them. The works that do 
consider general graph topologies usually take the following framework.
For a given function $f$ of interest, design first a protocol for 
securely computing $f$ in the simpler setting of a clique network,
then ``translate" this protocol 
to \emph{any} given graph $G$. 
Although this framework yields protocols that are secure in the strong sense (e.g., handling 
collusions and a malicious adversary), they do not quite fit the framework of 
distributed graph algorithms, and simulating these protocols in the  \congest\ 
model results in a large overhead in the round complexity. It is important to note that the blow-up in the 
number of rounds might occur regardless of the security requirement; for instance, when
the desired function $f$ is non-local, its distributed computation in general 
graphs might be costly with respect to rounds even in the \emph{insecure} 
setting. 
In the lack of distributed graph algorithms for general graphs that are both
secure and efficient compared to their \emph{non-secure} 
counterparts, we ask:
\begin{quote}
\begin{center}
{\em
How to design distributed algorithms that are both \textbf{\textit{efficient}} 
(in terms of round complexity) and \textbf{\textit{secure}} (where nothing is 
learned but the 
desired output)?}
\end{center}
\end{quote}
We address this challenge by introducing a new framework for secure distributed graph 
algorithms in the \congest\ model. Our approach is 
different from previous secure algorithms mentioned above and allows one to 
decouple between the price of locality and the price of security of a given 
function $f$. 
In particular, 
instead of adopting a clique-based secure protocol for $f$, we take the 
best distributed algorithm $\A$ for computing $f$, and then compile $\A$ 
to a secure algorithm $\A'$.
This compiled algorithm 
respects the same bandwidth limitations relies on no setup phase nor on any 
computational assumption and works for (almost) any graph.
The price of security comes as an overhead in the number of rounds.
Before presenting the precise parameters of the secure compiler, we first discuss the security notion used in this paper.

\paragraph{Our Security Notion.}
Consider a (potentially insecure) distributed algorithm $\A$.
Intuitively, we say that a distributed algorithm $\A'$ {\em securely} simulates 
$\A$ if (1) both algorithms have the exact same output for every node (or the exact 
same output distribution if the algorithm is randomized) and (2) each node 
learns ``nothing more'' than its final output. This strong notion of security is 
known as ``perfect privacy'' - which provides pure information theoretic 
guarantees
and relies on \emph{no computational assumptions}. The  perfect privacy notion 
is formalized 
by the existence of an 
(unbounded) simulator \cite{BenorGW88,GolVol2,Canetti00,AsharovL17}, with the 
following 
intuition: a node learns nothing 
except its own output $y$, from the messages it receives throughout the 
execution of the algorithm, if a simulator can produce the same output 
distribution while 
receiving only $y$ and the graph $G$.

Assume that one of the nodes in the network is an ``adversary'' that is trying to learn as 
much as possible from the execution of the algorithm. Then the security notion 
has some restrictions on the operations the adversary is allowed to perform: 
(1) The 
adversary 
is passive and only listens to 
the messages but does not deviate from the prescribed protocol; this is known in the literature as {\em semi-honest} security. (2) The 
adversary is {\em not} allowed to collude with other nodes in the network. As will be explained next, if the vertex connectivity of the graph is two, then this is the strongest adversary that one can consider. (3) 
The adversary gets to see the entire graph. That is, in this framework, the 
topology of the graph $G$ is not considered private and is not protected by the 
security notion. The private bits of information that are protected by our compiler are: 
the \emph{inputs} of the nodes (e.g., color) and the randomness 
chosen during the execution of the algorithm; as a result, the \emph{outputs} 
of the nodes are private (see \Cref{def:perfect-privacy} for precise 
details).

\paragraph{A Stronger Adversary: From Cliques to General Topology.}
The goal of this paper is to lay down the 
groundwork, especially the graph theoretic infrastructures for secure distributed graph algorithms. 
As a first step towards this goal, we consider the largest family of graphs for which secure computation can be achieved in the perfect secure setting. This is precisely the family of two-vertex connected graphs.

Handling this wide family of graphs naturally imposes restrictions on the power 
of the adversary that one can consider. In particular, we cannot hope to handle 
that standard adversary assumed in the MPC literature, which colludes with 
$\Omega(n)$ other parties. A natural limit on the adversarial collusion is the 
vertex connectivity of the graph. Indeed, if the graph is only $t$-vertex connected, 
then an adversary that colludes with $t$ nodes can receive all messages from 
one 
part of the graph to the other. That is, security for such a graph will imply 
a secure 
two-party computation, 
where each party simulates one connected component of the graph. Such two-party secure 
protocols where shown to be impossible for merely any interesting function
\cite{Kushilevitz89}. Thus, for 2-vertex connected graphs, one cannot achieve a 
secure simulation with an adversary that 
colludes with more than a node.
Moreover, the works of \cite{Dolev82} combined with \cite{DolevDWY93} show that 
if 
the adversary is {\em malicious} and colludes with $t$ nodes then graph must be 
$(2t+1)$ connected for security to hold. 

We believe that the framework provided in this paper, and the private neighborhood trees in particular, serve the basis for stronger security guarantees in the future, for highly connected graphs. 
  
\subsection{Our Results}
Our end result is the first \emph{general compiler} that can take any 
``natural'' 
(possibly insecure) distributed 
algorithm to one that has perfect security.
A ``natural'' 
distributed algorithm is one where the local computation at each node can be 
performed in 
polynomial time.
Through the paper, the $\widetilde{O}(\cdot)$ notation hides poly-logarithmic 
terms in the number of vertices $n$. Recall, that $G$ is 2-vertex 
connected (or bridgeless) if for all $u \in V$ the graph $G'=(V \setminus \{u\},E)$ is connected.

\begin{theorem}[Secure Simulation, Informal]\label{thm:secure-algorithm}
Let $G$ be a 2-vertex connected  $n$-vertex 
graph with diameter $D$ and maximal 
degree $\Delta$. Let 
$\mathcal{A}$ be a {\em natural} distributed algorithm that runs on 
$G$ in $r$ rounds. Then, $\mathcal{A}$ can be transformed to an equivalent 
algorithm $\mathcal{A}'$ 
with perfect privacy which runs in $\widetilde{O}(rD \cdot \poly(\Delta))$ 
rounds.
\end{theorem}
We note that our compiler works for any distributed algorithm rather than only 
on natural ones. The number of rounds will be proportional to the space 
complexity of the internal computation functions of the distributed algorithm (an explicit statement for any 
algorithm can be found in \Cref{remark:general}). 

This quite general framework is made possible due to fascinating
connections between ``secure cryptographic definitions"
and natural combinatorial graph properties. Most notably is a combinatorial 
structure that we call \emph{private 
neighborhood trees}. Roughly 
speaking, a private neighborhood tree 
of a 2-vertex connected graph $G=(V,E)$ is a collection of $n$ trees, one per node 
$u_i$, where each tree $T(u_i)\subseteq G \setminus \{u_i\}$ contains all the 
neighbors of $u_i$ but does not contain $u_i$. Intuitively, the private 
neighborhood trees allow all neighbors $\Gamma(u_i)$ of all nodes $u_i$ to 
exchange a secret without $u_i$. Note that these covers
exist if and only if the graph is 2-vertex 
connected. We define a $(\dilation,\congestion)$-private 
neighborhood trees in which each tree $T(u_i)$ has depth at most $\dilation$ 
and each edge belongs to at most $\congestion$ many trees. This allows the distributed compiler to 
use all trees simultaneously in $\widetilde{O}(\dilation+\congestion)$ rounds, by employing the random delay approach \cite{leighton1994packet,Ghaffari15}. 

\begin{theorem}[Private Neighborhood Trees]\label{thm:neighborhood-cover}
Every $2$-vertex connected graph with diameter $D$ and maximum 
degree $\Delta$ has a $(\dilation,\congestion)$-private neighborhood 
trees with $\dilation = \widetilde{O}(D \cdot \Delta)$ and $c = 
\widetilde{O}(D)$.
\end{theorem}
Our secure compiler assumes that the nodes in the graph know the private neighborhood trees. 
That is, each node knows its parent in the trees it participates in. Thus, 
in order for the whole transformation to work, the network needs to run a 
distributed algorithm to learn the cycle cover. This can be achieved by 
performing a prepossessing phase to compute the private trees and 
then run the distributed algorithm.


Since the barrier of \cite{Kushilevitz89} can be extended to a cycle 
graph, the linear dependency in $D$ in the round complexity of 
our compiler is unavoidable. 

The flow of our constructions is summarized in \Cref{fig:summary}.

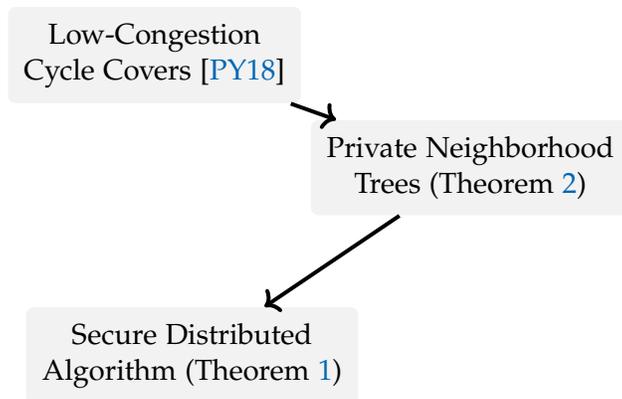
\begin{figure}[t!]
	\centering

	\begin{tikzpicture}
	[
	every node/.style={
		align = center,
		fill          =  black!5,
		inner sep     = 6pt,
		minimum width = 22pt,
		rounded corners=.1cm}
	]
	
	\node (cycle-cover) at (-6.5,0) {Low-Congestion\\Cycle Covers \cite{ParterY18}};
	\node (private-neighborhood-trees) at (-2.3,-1.5) {Private 
	Neighborhood\\Trees   
	(Theorem \ref{thm:neighborhood-cover})};
	\node (secure-algorithm) at (-6,-4) {Secure Distributed\\Algorithm 
	(Theorem~\ref{thm:secure-algorithm})};
	
	\draw[->, line width=1.5] (cycle-cover) -- (private-neighborhood-trees);
	
	\draw[->, line width=1.5] (private-neighborhood-trees) -- 
	(secure-algorithm);
	\end{tikzpicture}

	\caption{An illustrated summary of our results.}
	\label{fig:summary}
\end{figure}

\subsection{Applications for Known Distributed Algorithms}
\Cref{thm:secure-algorithm} enables us to compile almost all of the known 
distributed algorithms to a secure version of them.
It is worth noting that {\em deterministic} algorithms for problems in which 
the nodes do {\em not have any input} cannot be made secure by our approach 
since these algorithms only depend on the graph topology which we do not try to 
hide. 
Our compiler is meaningful for algorithms where the nodes have input or for 
randomized 
algorithms which define a distribution over the output of the nodes. For 
instance, the randomized 
coloring algorithms (see, e.g., \cite{barenboim2013distributed}) which sample a 
random legal coloring of 
the graph can be made secure.
Specifically, we get a distributed algorithm that 
(legally) colors a graph (or computes a legal configuration, in general), while 
the information that each node learns at the end is as if a centralized entity 
ran the algorithm for the entire network, and revealed each node's output 
privately (\ie revealing $v$ the final color of $v$). 

Our approach captures global (\eg MST) as well as many local problems 
\cite{naor1995can}. 
The MIS algorithm of Luby \cite{luby1986simple} along with our compiler yields 
$\widetilde{O}(D \cdot \poly(\Delta))$ secure algorithm according to the notion 
described above. 
Slight variations of this algorithm also give the $O(\log n)$-round 
$(\Delta+1)$-coloring algorithm (\eg Algorithm 19 of 
\cite{barenboim2013distributed}). Combining it with our compiler, we get 
a secure $(\Delta+1)$-coloring algorithm with round complexity of 
$\widetilde{O}(D \cdot 
\poly(\Delta))$.
Using the Matching algorithm of Israeli and Itai \cite{israeli1986fast} we get 
an $\widetilde{O}(D\cdot \poly(\Delta))$ secure maximal matching algorithm.  
Finally, another example comes from distributed algorithms for the Lov{\'a}sz 
local lemma (LLL) which have received much attention recently 
\cite{brandt2016lower,FischerG17,chang2017time} for the class of bounded degree 
graphs. 
Using \cite{chung2017distributed}, most of these (non-secure) algorithms for 
defective coloring, frugal coloring, and list vertex-coloring can be made 
secure within $\widetilde{O}(D)$ rounds.

\subsection{Related Work}
There is a long line of research on secure multiparty computation.
The most general results are \cite{Yao82b,GoldreichMW87,BenorGW88,ChaumCD88} 
which provide a protocol for computing any function $f(x_1,\ldots,x_n)$ over 
$n$ inputs. The significance of these results is that they work for any 
function and provide strong security (either assuming computational 
assumptions, or information-theoretic as is considered in this work). More 
specifically, these results can tolerate the adversary colluding with $t$ nodes 
as long as $t < n/2$ and achieve security in the semi-honest model. If the 
adversary is malicious (\ie she can deviate from the prescribed protocol) then 
the protocols can handle $t < n/3$. These general results assume that every two 
parties have a secure direct channel between them, that is, they assume that 
the interaction pattern is a clique.

\paragraph{General Graph Topologies.}
As opposed to general MPC, there are not many protocols that work for general 
graph interaction patterns. The works of \cite{HaleviIJKR16,HaleviI0KSY17} 
provide secure protocols for any function $f$ and for any graph pattern, 
however, they have a few drawbacks. First, they both assume some form of a 
setup phase (recall that our solution assumes no such setup). The work of 
\cite{HaleviIJKR16} assumes a setup of correlated randomness and provides 
several protocols with information-theoretic security. Even the most efficient 
protocol (one for symmetric functions) must send more than $n^2$ bits on each 
edge. In the CONGEST model, this would take $n^2/\log n$ rounds which is not an
interesting regime for this model (for other functions the communication 
complexity is much worse). The work of \cite{HaleviI0KSY17} assumes a 
public-key infrastructure (PKI) setup together with a common random string, and 
moreover provides only solutions based on (heavy) computational assumptions. 
More importantly, the number of rounds is not a parameter they optimize and 
will be at least $O(n^2)$ even for simple functions. Such protocols were also 
implicitly considered in \cite{GoldwasserGG0KLSSZ14,BeimelGIKMP14} but suffer 
from similar drawbacks.

\paragraph{MPC with Locality Constraints.} 
Other MPC works study protocols with small locality, that is, that parties 
communicate only with a small number of other parties. The work of 
\cite{BoyleGT13} (followed by \cite{ChandranCGGOZ15}) provides a secure MPC 
protocol for general functions where for $n$ parties each party is only 
required to communicate with at most $\polylog(n)$ other parties using 
$\polylog(n)$ rounds, where in each round messages are of size $\poly(n)$. 
Their protocols assume computation assumptions and a setup phase. Moreover, we 
stress that they still require a \emph{fully} connected network, and then 
parties communicate with a small number of dynamically chosen parties, but with 
very large messages (compared to the $\log n$ bound in our model). On the 
positive side, they achieve a strong security notation where the adversary can 
collude with up to about $n/3$ nodes. 

\paragraph{MPC for Bounded Degree Graphs.} 
The works of \cite{GarayO08,ChandranGO10} consider MPC protocols for bounded 
degree graphs. Except for the graph restriction, they consider arbitrary 
interaction pattern, information-theoretic security and no setup phase (similar 
to our setting). They provide a protocol for computing any function $f$ while 
``giving up'' on security for some nodes (\ie the adversary might learn the 
private inputs of these nodes). While they have some restrictions on the 
adversary, their security notion is stronger than ours (as we do not allow 
collusions). However, their protocol is not round nor bandwidth efficient with 
respect to our parameters. They simulate the general MPC protocol of 
\cite{BenorGW88}, by replacing each message sent from player $i$ to player $j$, 
one-by-one, with a distributed broadcast protocol that takes more than $n$ 
rounds by itself. Thus, the overall protocols will require at least $n^2$ 
rounds and enters the uninteresting regime of parameters for the CONGEST model.

\paragraph{The Key Differences to Our Approach.}
Our approach is quite different from the algorithms mentioned above. In particular, 
instead of taking some 
function $f$ and trying to build the best secure protocol for it, our compiler 
takes the best \emph{distributed} algorithm for computing $f$, and then compiles it 
to a secure one. This makes us competitive, in terms of the number of rounds, with the non-secure distributed algorithm. In the MPC approach, it is harder to decouple between the price of locality and the price of security as adopting a clique-protocol for computing a function $f$ to a general graph might blow up the number of rounds, regardless of the security constraint. 
One exception is the work of \cite{KearnsTW07} who showed how to compile a distributed algorithm for a {\em specific} task of Belief Propagation to a secure one. The compiler is designed specifically for this task and does not work for others. Moreover, the security that is achieved is based on computational assumption and specifically public-key encryption with additional properties. 

Finally, we note that a compiler that works for even a weaker adversary was 
proposed in \cite{ParterY18}. In their setting, the adversary can listen to the 
messages of a single \emph{edge} where in our setting the adversary listens to 
all messages received by a single \emph{node}. Thus, the adversary gets 
messages from $\Delta$ different edges in addition to getting the private 
randomness chosen by the node itself.
\section{Our Approach and Techniques}
We next describe the high-level ideas of our secure compiler. In 
\cref{sec:security}, we describe how the secure computation in the distributed 
setting boils down into a novel graph theoretic structure, namely, private 
neighborhood trees. The construction of private trees is shown in 
\Cref{subsec:privaencover}.

\subsection{From Security Requirements to Graph Structures}\label{sec:security}
In this section, we give an overview of the construction of a secure compiler 
and begin by showing how to compile a single round distributed algorithm into a 
secure one. This single-round setting already captures most of the challenges 
of the compiler. At the end of the section, we describe the additional ideas 
required for 
generalizing this to arbitrary $r$-round algorithms.

\paragraph{Secure Simulation of a Single Round.}
Let $G$ be an $n$-vertex graph with maximum degree $\Delta$, and for any node 
$u$ let 
$\sigma_u$ be its initial state (including its ID, and private input). Any 
single round algorithm can be described by a function $f$ that maps the initial 
state $\sigma_u$ of $u$ and the messages $m_1,\ldots, m_{\Delta}$ received from 
its neighbors to the output of the algorithm for the node 
$u$. 

As a concrete running example, let $\A$ be a single round algorithm that 
verifies vertex coloring in a graph. In this algorithm, the initial state of a node includes 
a color $c_u$, and nodes exchange their color with their neighbors and output 
$1$ if and only if all of their neighbors are colored with a color that is 
different from $c_u$.
It is easy to see that in this simple algorithm, the nodes learn more than the 
final output of the algorithm, namely, they learn the color of their neighbors.
Our goal is to compile this algorithm to a secure one, where nothing is learned 
expect the final output. In particular, where nodes do not learn the color of 
any other node in the network. This fits the model of Private 
Simultaneous Messages (PSM) that we describe next. We stress that other MPC protocols might be suitable here as well (\eg \cite{Yao82b,GoldreichMW87,BenorGW88}), 
however, the star topology of PSM model makes the best fit in terms of simplicity and parameters.

The PSM model was introduced by Feige, Kilian and Naor \cite{FeigeKN94} (and 
later generalized by \cite{IshaiK97}) as a ``minimal'' 
model of MPC for securely computing a function $f$. In this model, there are 
$k$ clients
that hold inputs $x_1,\ldots, x_k$ which are all connected to a single server 
(i.e., a star topology). The clients share private randomness $R$ that is 
hidden 
from the server. The goal is for the server to compute $f(x_1,\ldots, 
x_k)$ while learning ``nothing more'' but this output. The protocol consists 
of a single round where each client $i$ sends a message to the server that 
depends on its own input $x_i$ and the randomness $R$. The server, using these 
messages, computes the final output $f(x_1,\ldots, 
x_k)$. In \cite{FeigeKN94}, it was shown that any function $f$ admits such a 
PSM protocol with information-theoretic privacy. The complexity measure of the 
protocol is the size of the messages (and shared randomness) which are 
exponential in the space complexity of the function $f$ (see 
\Cref{def:psm} and \Cref{thm:psm} for precise details).

Turning back to our single round distributed algorithm $\cA$, the secure simulation of $\cA$ can be based 
on the PSM protocol for securely computing the function $f$, the  function that characterizes algorithm $\A$. In this view, each node $u$ in the 
graph acts as a server in the PSM protocol, while its (at most $\Delta$) neighbors in the graph act as the
clients. 

In order to simulate the PSM protocol of \cite{FeigeKN94} in the \congest\ model, one has to take care of several 
issues. The first issue concerns the bandwidth restriction; in the \congest\ 
model, every neighbor $v_i$ can send only $O(\log n)$ 
bits to $u$ in a single round. Note that the PSM messages are exponential in 
the space complexity of the function $f$, and that in our setting the total 
input of $f$ has $O(\Delta\log n)$ bits. Thus,
in a na\"ive implementation only functions $f$ that are computable in 
logarithmic 
space can be computed with the desired overhead of $\poly(\Delta)$ rounds. 
Our goal is to capture a wider family of functions, in particular, the class of 
natural algorithms in which $f$ is computable in polynomial time. Therefore, in 
our final compiler, we do not compute $f$ in a single round, but rather compute it gate-by-gate. Since 
in natural algorithms $f$ is computed by a circuit of polynomial size, and 
since a single gate is computable in 
logarithmic space, we incur a total round overhead that is polynomial in 
$\Delta$.
In what follows, assume that $f$ is computable 
in logarithmic space.

Another issue to be resolved is 
that in the PSM model, the server did not hold an input whereas in our 
setting the function $f$ depends not only on the input of the neighbors but on 
the input of the node $u$ as well. This subtlety is handled by having $u$ 
secret share\footnote{A secret share of $x$ to $k$ parties is a random tuple 
$r_1,\ldots,r_k$ such that $r_1 \oplus \dots \oplus r_k=x$.} its input to the 
neighbors.

\paragraph{How to Exchange Secrets in a Graph?}
There is one final critical missing 
piece that requires hard work: the neighbors of $u$ must share private 
randomness $R$ that is {\em not} known to $u$.
Thus, the secure simulation of a single round distributed algorithm can be 
translated into the following problem:
\begin{quote}
\begin{center}
{\em
How to share a secret $R_u$ between the neighbors of each node $u$ in the graph 
while hiding it from $u$ itself?}
\end{center}
\end{quote}
Note that this task should be done for all nodes $u$ in the graph $G$ 
simultaneously. That is, for every node $u$, we need the neighbors of $u$ to 
share a private random string that is 
\emph{hidden} from $u$. Our solution to this problem is information 
theoretic and builds upon specific graph structures. However, we begin by 
discussing a much simpler solution, yet, based on computational assumptions.

\paragraph{A Solution Based on Computational Assumptions.}
In order to get a computationally based solution, we assume the existence of a public-key encryption scheme.
For simplicity, we assume that our public-key encryption scheme has two properties: (1) the encryption does 
not increase the size of the plaintext, and (2) the length of the public-key is 
$\lambda$ -- the security 
parameter of the public-key scheme. 
We next describe an $\widetilde{O}(\Delta + \lambda)$ protocol that computes the secret $R$ which is shared by all neighbors of $u$ while hiding it from $u$, under the public-key assumption. 

Consider a node $u$ and let $v_1,\ldots,v_{\Delta}$ be its neighbors. For 
simplicity, assume that $\Delta$ is a power of $2$.
First, $v_1$ computes the random string $R$, this string will be shared with 
all $v_i$'s nodes in $\log \Delta$ phases. In each phase $i\geq 0$, we assume 
that all the vertices $v_{1},\ldots, v_{k_i}$ for $k_i=2^{i}$ know $R$. We will 
show that at the end of the phase, all vertices $v_{1},\ldots, v_{k_i}, 
v_{k_i+1},\ldots, v_{2k_i}$ know $R$. This is done as follows. Each vertex 
$v_{k_i+j}$ sends its public-key to $v_{j}$ via the common neighbor $u$, 
$v_{j}$ encrypts $R$ with the key of $v_{k_i+j}$ and sends this encrypted 
information to $v_{k_i+j}$ via $u$. As the length of the public-key is 
$\lambda$ and the length of the encrypted secret $R$ needed by the PSM protocol has $O(\Delta \log n)$ bits,
this can be done in $\widetilde{O}(\Delta + \lambda)$ rounds. It is 
easy to see that $u$ cannot learn the secret $R$ under the public-key 
assumption.  
%

Using this protocol with the PSM machinery yields a protocol that compiles any $r$-round algorithm $\cA$ (even non-natural one) into a secure algorithm $\cA$ with $r'=\widetilde{O}(r(\Delta + \lambda))$ rounds.
We note that it is not clear what is $\lambda$ as a function of the number of nodes $n$. 
Clearly, if $\lambda=\Omega(n)$, this overhead is quite large. The benefit of our perfect security scheme is that 
it relies on no computational assumptions, does not introduce an additional 
security parameter and as a result, the round complexity of the compiled 
algorithms depends only on 
the properties of the graph, e.g., number of nodes, maximum degree and 
diameter. Finally, the dependency on these graph parameters is existentially 
required.

\paragraph{Our Information-Theoretic Solution.}
Suppose two nodes, $v_1,v_2$, wish to share information that is hidden from a 
node $u$ in the information-theoretic 
sense. Then, they must use a $v_1$-$v_2$ path in $G$ that is \emph{free} from 
$u$. 
Hence, in order for the neighbors of a node $u$ to exchange private randomness, 
they must use a connected subgraph $H$ of $G$ that spans all the neighbors of $u$ but
does not include $u$. (This, in particular, explains our requirement for $G$ to 
be 
2-vertex connected.) Using this subgraph $H$, the neighbors can communicate 
privately (without $u$), and exchange shared randomness.

In order to reduce the 
overhead of the compiler, we need the diameter of $H$ to be as small as 
possible. Moreover, in the compiled algorithm, we will have the neighbors of 
all nodes $u$ in the graph exchange randomness simultaneously. Since there is 
a bandwidth limit, we need to have minimal overlap of the different subgraphs 
$H$. It is easy to see that for every vertex $u$, there exists a tree $T(u) \subseteq G\setminus \{u\}$ of diameter $O(\Delta \cdot D)$ that spans all the neighbors of $u$. However, an arbitrarily collection of trees $T(u_1),\ldots, T(u_n)$ where each $T(u_i)\subseteq G\setminus \{u_i\}$ might result in an edge that is common to $\Omega(n)$ trees. This is undesirable as it might lead to a blow-up of $\Omega(n)$ in the round complexity of our compiler.

Towards this end, we define the notion 
of private neighborhood trees which provides us the communication backbone for 
implementing this distributed $\PSM$ protocol in general graph 
topologies for all nodes simultaneously.
Roughly 
speaking, a private neighborhood tree 
of a 2-vertex connected graph $G=(V,E)$ is a collection of $n$ trees, one per node 
$u_i$, where each tree $T(u_i)\subseteq G \setminus \{u_i\}$ contains all the 
neighbors of $u_i$ but does not contain $u_i$. A 
$(\dilation,\congestion)$-private neighborhood trees in which each tree 
$T(u_i)$ has depth at most $\dilation$ 
and each edge belongs to at most $\congestion$ many trees.
This allows us to 
use all trees simultaneously and exchange all the private randomness in 
$\widetilde{O}(\dilation+\congestion)$ rounds.

Let $G$ be a $2$-vertex connected graph and let $D$ be the diameter of $G$. 
By the discussion above, achieving $(\dilation,\congestion)$-private neighborhood trees
with $\dilation=O(\Delta \cdot D)$ and $\congestion=n$ is easy, but 
yields an inefficient compiler.
We show how to construct 
$(\dilation,\congestion)$-private neighborhood trees for $d = \widetilde{O}(D 
\cdot \Delta)$ and $\congestion = \widetilde{O}(D)$, these parameters are nearly optimal existentially. The construction builds on 
a simpler and more natural structure called cycle cover. Using these private neighborhood trees, the neighbors of each node $u$ can exchange the $O(\Delta \log n)$ bits of $R_u$ in $\widetilde{O}(\Delta \cdot D)$ rounds.  This is done for all nodes $u$ simultaneously using the random delay approach.

Note that unlike the computational setting, here the round complexity is existentially optimal (up to poly-logarithmic terms) and only depends on the parameters of the graph.


\paragraph{Secure Simulation of Many Rounds.}
We have described how to securely simulate single round distributed 
algorithms. Consider an $r$-round distributed algorithm $\mathcal{A}$. In a 
broad view, $\mathcal{A}$ can be viewed as a collection of $r$ functions 
$f_1,\ldots,f_r$. At round $i$, a node $u$ holds a state 
$\sigma_i$ and needs to update its state according to a function $f_i$ that 
depends on $\sigma_{i}$ and the messages it has received in this round. 
Moreover, the same function $f_i$ computes the messages that $u$ will send to its neighbors in the next round.
That is, 
$$f_i(\sigma_{i}, m_{v_1 \to u}, \ldots, m_{v_\Delta \to u}) = 
(\sigma_{i+1},m_{u \to v_1},\ldots,m_{u \to v_{\Delta}})~.
$$
Assume 
that the final state $\sigma_r$ is the final output of the algorithm for node 
$u$.
A first attempt is to simply apply the solution for a single 
round for $r$ many times, round by round. As a result, the node $u$ learns all internal states 
$\sigma_1,\ldots,\sigma_r$ and nothing more. This is, of course, undesirable as 
these internal states, $\sigma_i$ for $i\leq r-1$, 
might already reveal much more information than the final output. Instead, 
we simulate the computation of the internal states
$\sigma_1,\sigma_2,\ldots,\sigma_r$, 
in an \emph{oblivious} manner without knowing any $\sigma_i$ except for 
$\sigma_r$ which is the final output of the algorithm.

Towards this end, in our scheme, the node $u$ holds an ``encrypted'' state, 
$\widehat{\sigma}_i$, instead of the actual state $\sigma_i$. The encryption 
we use is a simple one-time-pad where the key is a random string 
$R_{\sigma_{i}}$ such 
that $\widehat{\sigma}_i \oplus R_{\sigma_{i}} = \sigma_i$. The key 
$R_{\sigma_{i}}$ will be 
chosen by an arbitrary neighbor $v$ of $u$. In addition to the state, the node 
$u$ should not be able to learn the messages $m_{u \to v_1},\ldots,m_{u \to 
v_{\Delta}}$ sent to the neighbors in the original algorithm. Thus, each 
neighbor $v_j$ holds the key $R_{u \to v_j}$ that is used to encrypt its 
incoming message to $u$.
Overall, at any given round $i$, any node $u$ holds an encrypted 
state $\widehat{\sigma}_i$, and encrypted outgoing messages $\widehat{m}_{u \to 
v_1},\ldots,\widehat{m}_{u \to v_\Delta}$; the neighbors of $u$ hold the 
corresponding decryption keys. To compute the new state and the messages that $u$ sends to its neighbors in the 
next round, we use the PSM protocol as described in a single round but with 
respect to a function $f'_i$ which is related to the function $f_i$ and 
is defined as follows. The input of the function $f'_i$ is an encrypted state (of $u$), encrypted messages from its neighbors, keys for 
decrypting the input, and new keys for encrypting the final output. First, the function $f'_i$ 
decrypts the encrypted input to get the original state and the messages sent from its neighbors (i.e., the input for function $f_i$). 
Then, the function $f'_i$ applies $f_i$ to get the next state $\sigma_{i+1}$ and new outgoing 
messages from $u$ to its neighbors. Finally, it uses new encryption keys to encrypt the new output and 
finally outputs the new encrypted data (states and messages to be sent).
A summary of the algorithm for a single node $u$ is given in 
\Cref{fig:alg-A_u-intro}. The full proof is given in \Cref{sec:simulation}.

\begin{figure}[!h]
\begin{boxedminipage}{\textwidth}
\vspace{1mm} \textbf{Algorithm $\A_u$:}
\begin{enumerate}
	\item For each round $i=1\ldots r$ do:
	\begin{enumerate}
	\item $u$ holds the encrypted state  $\widehat{\sigma_i}$.
	\item The neighbor $v$ of $u$ samples new encryption keys.
	\item Run a $\PSM$ protocol with $u$ as the server to compute the 
	function $f'_i$:
	\begin{enumerate}
		\item $u$ sends its state $\widehat{\sigma_i}$ to a neighbor $v' \ne v$.
		\item Neighbors share private randomness via the \textbf{private 
		neighborhood trees}.
		\item $u$ learns its new encrypted state $\widehat{\sigma}_{i+1}$.
	\end{enumerate}	
	\end{enumerate}
	\item $v$ sends the final encryption key to $u$.
	\item Using this key, $u$ computes its final output $\sigma_r$.
\end{enumerate}
\end{boxedminipage}
\caption{The description of the simulation of algorithm $\mathcal{A}$ with 
respect to a node $u$.}
\label{fig:alg-A_u-intro}
\end{figure}

\subsection{Constructing Private Neighborhood Trees}
Our construction of private neighborhood trees is based on the construction of 
{\em low-congestion cycle-covers} from \cite{ParterY18}. For a bridgeless graph $G=(V,E)$, a low congestion cycle 
cover is a decomposition of graph edges into cycles which are both short and 
almost \emph{edge-disjoint}. Formally, a $(\dilation,\congestion)$-cycle cover 
of a graph $G$ is a collection of cycles in $G$ in which each cycle is of 
length at most $\dilation$, and each edge participates in at least one cycle 
and at most $\congestion$ cycles. In \cite{ParterY18} the following theorem was 
proven:

\begin{theorem}[Low Congestion Cycle Cover, \cite{ParterY18}]\label{thm:cyclecover}
Every bridgeless graph\footnote{A graph $G=(V,E)$ is bridgeless if $G\setminus \{e\}$ is connected for every $e$.} with diameter $D$ has a $(\dilation,\congestion)$-cycle cover where $\dilation=\widetilde{O}(D)$ and 
$\congestion = \widetilde{O}(1)$.
That is, the edges of $G$ can be covered by cycles such that each cycle is of 
length at most $\widetilde{O}(D)$ and each edge participates in at most 
$\widetilde{O}(1)$ cycles.
\end{theorem}
To prove \Cref{thm:neighborhood-cover} we show how to use the construction of a 
 $(\dilation,\congestion)$-cycle cover $\cC$ to obtain a $(\dilation \cdot \Delta,\congestion \cdot D \cdot \log \Delta)$ private neighborhood trees $\cN$. Using the construction of $(\widetilde{O}(D),\widetilde{O}(1))$ cycle covers of \Cref{thm:cyclecover} yields the theorem.

Consider a node $u$ with only two neighbors $v_1,v_2$. Then, a cycle cover of 
the graph must cover the edge $(u,v_1)$ by using a cycle containing the node 
$v_2$. Thus, the cycle induces a path between $v_1$ and $v_2$ that does not 
contain $u$. We get a private neighbor tree for $u$, a short path from $v_1$ to 
$v_2$ that does not go through $u$ and has low congestion. We use this idea, 
and show how to generalize it to nodes with an arbitrary degree.

The construction of the private neighborhood trees $\cN$ consists of 
$O(\log \Delta)$ phases. In each phase, we compute a low-congestion
cycle cover in some auxiliary graph using \Cref{thm:cyclecover}. We begin
by each node $u$ holding an empty forest $F_0(u)$ 
consisting only of $u$'s neighbors. We then compute a cycle cover in the graph 
$G$. Let $v_1,\ldots,v_\Delta$ be the neighbors of $u$. Then, the cycles of the 
cycle cover provide short paths between pairs $(v_i,v_j)$ that avoid $u$. We 
add these paths to the forest of $u$. By doing this, we have reduced the number of 
connected components in the forest of $u$ by half. Importantly, since the added
paths are obtained from low-congestion cycle cover, adding these paths for all nodes 
$u$ in the graph still keeps the congestion per edge bounded.

The high-level idea is to repeat this process for $\log \Delta$ iterations 
until $u$'s forest contains one connected component: the output private 
tree $T(u)$ for the node $u$. In order to run the next iteration, we must force 
the cycle cover to find different cycles than the ones it has computed in the 
previous iteration.

Towards that goal, the algorithm uses an auxiliary graph $G'$ defined as follows. First, we add all the nodes in $G$ to $G'$ (but not the edges).
Consider the collection of connected components of a node $u$. We add a virtual node $u_j$ to $G'$ for each of the connected components of $u$, and connect $u_j$ to $u$. Finally, every edge $(u,v)$ in $G$ is replaced by an edge $(u_j,v_i)$ in $G'$ where $v$ is in the $\ith{j}$ connected component of $u$ and $u$ is in the $\ith{i}$ connected component of $v$. See \Cref{fig:trees} for an  illustration.

\begin{figure}[h!]
	\begin{center}
\includegraphics[scale=0.50]{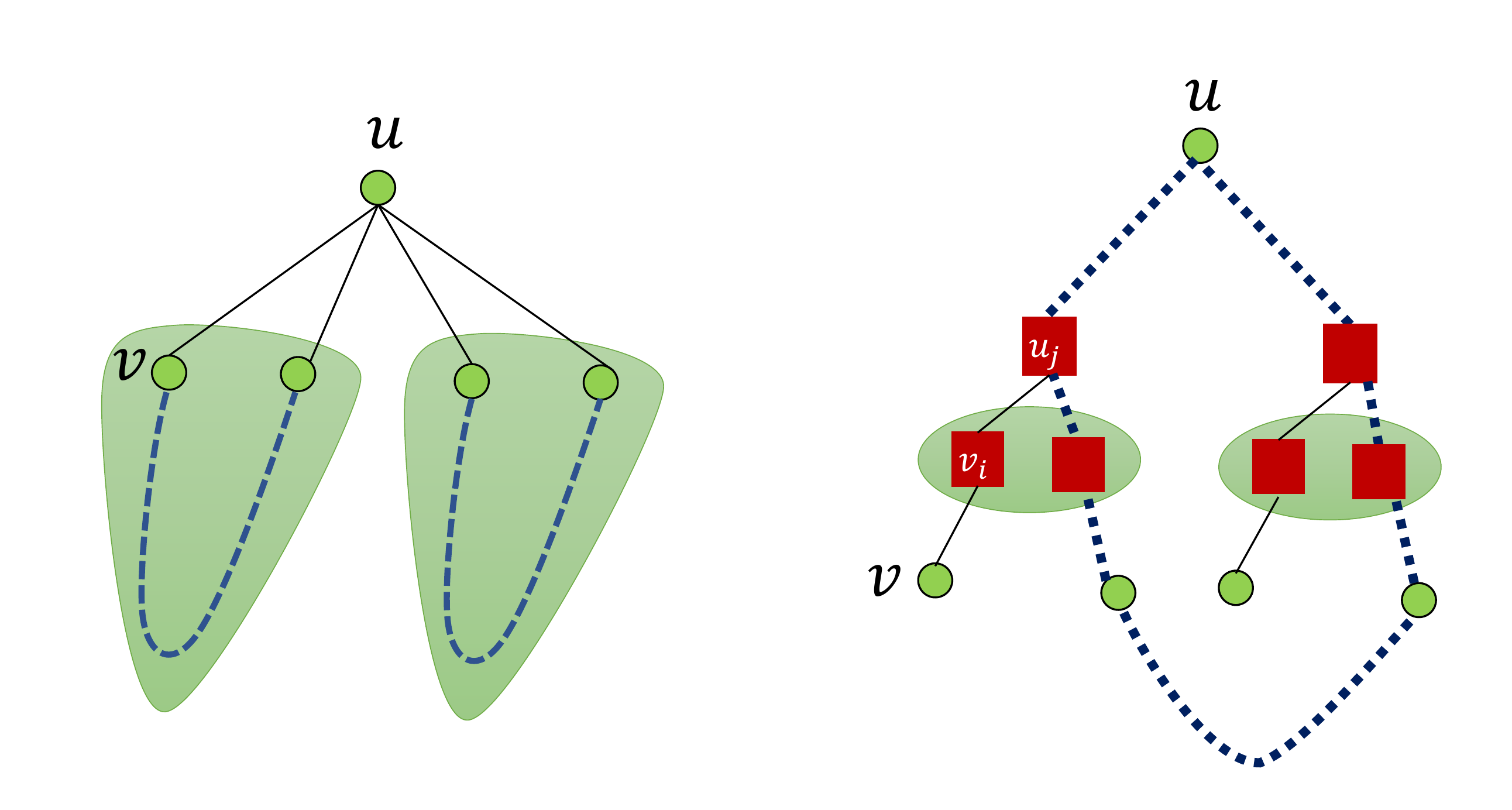}
\caption{Left: Illustration of a node $u$ with two connected components in the original graph $G$. Right: The auxiliary graph $G'$ and a new cycle found by the cycle cover in $G'$ (dashed line). One can observe how this new cycle now connects the two connected components.
	\label{fig:trees}}
	\end{center}
\end{figure}

Now, we run the cycle cover algorithm on $G'$. In the analysis section we show that since the graph $G$ is $2$-vertex connected, the graph $G'$ is also $2$-vertex connected.  A cycle covering an edge $(u,u_j)$ in $G'$ must use the virtual edge $(u,u_i)$ for $i \ne j$. That is, this cycle is used to obtain a path between the $\ith{j}$ connected component and the $\ith{i}$ connected component of $u$ while avoiding $u$. Adding these paths will again reduce the number of connected components by a half. Finally, since these paths are computed in a virtual graph $G'$, we need to translate them into $G$-paths. This is done as follows. An edge $(u,u_j)$ is simply replaced by $u$. An edge $(u_j,v_i)$ is replaced by the edge $(u,v)$ in $G$. We then use the fact that 
the cycles computed in $G'$ have low congestion in $G'$, to show that also the translated paths have low congestion in $G$.

To bound the depth of the private tree $T(u)$, we note that the final tree has 
at most $\Delta$ neighbors of $u$. Each tree path connecting two consecutive 
$u$'s neighbors is obtained from a cycle cover computation and thus have length 
$O(D \log n)$. Overall, this gives a diameter of $O(\Delta D \log n)$. 
The detailed description and analysis of this construction are provided in \Cref{subsec:privaencover}.
\section{Preliminaries}\label{sec:prelims}
Unless stated otherwise, the logarithms in this paper are base 2.
For an integer $n \in \mathbb{N}$ we denote by
$[n]$ the set $\{1,\ldots, n\}$. We denote by $U_n$ the uniform distribution
over $n$-bit strings. For two distributions (or random variables) $X,Y$ we 
write $X \equiv Y$ if they are identical distributions. That is, for any $x$ it 
holds that $\Pr[X=x]=\Pr[Y=x]$.

\paragraph{Graph Notations.} 
For a tree $T \subseteq G$, let $T(z)$ be the 
subtree of $T$ rooted at $z$. 
Let $\Gamma(u,G)$ be the neighbors of $u$ in $G$, and $\deg(u,G)=|\Gamma(u,G)|$. When $G$ is clear from the context, we simply write $\Gamma(u)$ and $\deg(u)$. 
%

\subsection{Distributed Algorithms}

\paragraph{The Communication Model.}
We use a standard message passing model, the
\congest\ model \cite{Peleg:2000}, where the execution proceeds in synchronous 
rounds and in each round, each node can send a message of size $O(\log n)$ to 
each of its neighbors.
In this model, local computation at each node is for free and the primary 
complexity measure is the number of communication rounds. Each node holds a 
processor with a unique and arbitrary ID of $O(\log n)$ bits.
Throughout, we make extensive use of the following useful tool, which is based 
on the random delay approach of \cite{leighton1994packet}. 
\begin{theorem}[{\cite[Theorem 1.3]{Ghaffari15}}]\label{thm:delay}
Let $G$ be a graph and let $A_1,\ldots,A_m$ be $m$ distributed algorithms in 
the \congest model, 
where each algorithm takes at most $\dilation$ rounds, and where for each 
edge of $G$, at most $\congestion$ messages need to go through it, in total 
over all these algorithms. Then, there is a randomized distributed
algorithm (using only private randomness) that, with high probability, produces 
a schedule that runs all the algorithms in $O(\congestion +\dilation \cdot \log 
n)$ rounds, after $O(\dilation \log^2 n)$ rounds of pre-computation.
\end{theorem}

\paragraph{A Distributed Algorithm.}
Consider an $n$-vertex graph $G$ with maximal degree $\Delta$.
We model a distributed algorithm $A$ that works in $r$ rounds as describing 
$r$ functions $f_1,\ldots,f_r$ as follows. Let $u$ be a node in the graph with 
input $x_u$ and neighbors $v_1,\ldots,v_{\Delta}$. At any round $i$, the 
memory of a node $u$ consists of a state, denoted by $\sigma_i$ and $\Delta$ 
messages $m_{v_1 \to u}\ldots, m_{v_\Delta \to u}$ that were received in 
the previous round.

Initially, we set $\sigma_0$ to contain only the input $x_u$ of $u$ and its 
ID and initialize all messages to $\bot$. At round $i$ the node $u$ updates its 
state to $\sigma_{i+1}$ according to its previous state $\sigma_i$ and the 
message from the previous round, and prepares $\Delta$ messages to send $m_{u 
\to v_1}, \ldots, m_{u \to v_{\Delta}}$.
To ease notation (and without loss of 
generality) we assume that each state contains the ID of the node $u$, thus, we 
can focus on a single update function $f_i$ for every round that works for all 
nodes. The function $f_i$ gets the state $\sigma_i$ and messages $m_{v_1 \to 
u}\ldots, m_{v_\Delta \to u}$, and randomness $s_i$ and outputs the next 
state and outgoing message:
\begin{align*}
	(\sigma_{i},m_{u \to v_1}, \ldots, m_{u \to v_\Delta}) \gets 
	f_i(\sigma_{i-1},m_{v_1 \to u},\ldots, m_{v_\Delta \to u}, s_i).
\end{align*}
At the end of the $r$ rounds, each node $u$ has a state $\sigma_r$ and a final 
output of the algorithm. Without loss of generality, we assume that 
$\sigma_r$ is the final output of the algorithm (we can always modify $f_r$ 
accordingly).

\paragraph{Natural Distributed Algorithms.}
We define a family of distributed algorithms which we call {\em natural}, which 
captures almost all known distributed algorithms. A 
natural distributed algorithm has two restrictions for any round $i$: (1) the 
size the state is bounded by $|\sigma_i| \le \Delta \cdot \polylog(n)$, and (2) 
the function $f_i$ is computable in polynomial time. 
The input for $f_i$ is the state $\sigma_i$ and at most $\Delta$ message each 
of length $\log n$. Thus, the input length $m$ for $f_i$ is bounded by $m \le 
\Delta \cdot \polylog(n)$, and the running time should be polynomial in this
input length.

We introduce this family of algorithms mainly for simplifying the presentation of 
our main result. For these algorithms, our main statement can be described with minimal overhead. However, our results are general and 
work for any algorithm, with appropriate dependency on the size of the state 
and the running time the function $f_i$ (i.e., the internal computation time at each node $u$ in round $i$).

\paragraph{Notations.}
We introduce some notations: For an algorithm $\A$, graph $G$, 
input $X=\{x_v\}_{v \in G}$ we denote by $\A_u(G,X)$ the random variable of the 
output of node $u$ while performing algorithm $\A$ on the graph $G$ with inputs 
$X$ (recall that $\A$ might be randomized and thus the output is a random 
variable and not a value). Denote by $\A(G,X)=\{\A_u(G,X)\}_{u \in G}$ the 
collection of outputs (in some canonical ordering). Let $\View_u^{\A}(G,X)$ be a 
random variable of the viewpoint of $u$ in the running of the algorithm $\A$. 
This includes messages sent to $u$, its memory and random coins during all 
rounds of the algorithm.

\paragraph{Secure Distributed Computation.}
Let $\A$ be a distributed algorithm. Informally, we say that $\A'$ simulates $\A$ 
 in a secure manner if when running the algorithm $\A'$ 
every node $u$ only learns its final output in $\A$ and ``nothing more''. 
This notion is captured by the existence of a simulator and is defined below.

\begin{definition}[Perfect Privacy]\label{def:perfect-privacy}
Let $\A$ be a distributed (possibly randomized) algorithm, that works in $r$ 
rounds.
We say that an algorithm $\A'$ simulates $\A$ with perfect privacy if for every 
graph $G$, every $u \in G$ and it holds that:
\begin{enumerate}
\item \textbf{Correctness:} For every input $X=\{x_v\}_{v \in V}$: $\A(G,X) 
\equiv \A'(G,X)$.
\item \textbf{Perfect Privacy:} There exists a randomized algorithm (simulator) 
$\Sim$ such that for every input $X=\{x_v\}_{v \in V}$ it 
holds that $$\View_u^{\A'}(G,X) \equiv \Sim(G,x_u,\A_u(G,X)).$$
\end{enumerate}
\end{definition}
This security definition is known as the ``semi-honest'' model, where the 
adversary, acting as one of the nodes in the graph, is not allowed to deviate 
from the prescribed protocol but can run arbitrary computation given all the 
messages it received. Moreover, we assume that the adversary does no collude 
with other nodes in the graph.

\subsection{Cryptography with Perfect Privacy}
One of the main cryptographic tools we use is a specific protocol for secure 
multiparty computation that has perfect privacy.
Feige Kilian and Naor~\cite{FeigeKN94} suggested a model where two players 
having inputs $x$ and $y$ wish to compute a function $f(x,y)$ in a secure 
manner. They achieve this by each sending a single message to a third party 
that is able to compute the output of the function $f$ from these messages but 
learn nothing else about the inputs $x$ and $y$. For the protocol to work, the 
two parties need to share private randomness that is not known to the third 
party. This model was later generalized to 
multi-players and is called the Private Simultaneous Messages Model 
\cite{IshaiK97}, 
which we formally describe next.
\begin{definition}[The $\PSM$ model]\label{def:psm}
Let $f \colon (\bit^m)^k \to \bit^m$ be a $k$ variant 
function. A $\PSM$ protocol for $f$ consists of a pair of algorithms 
$(\PSMEnc,\PSMDec)$ where $\PSMEnc \colon  \bit^m \times \bit^{r} \to \bit^t$ 
and 
$\PSMDec \colon (\bit^t)^k \to \bit^m$ such that
\begin{itemize}
\item For any $X=(x_1,\ldots,x_k)$ it holds that: 
$\Pr_{R \in 
\bit^{r}}[\PSMDec(\PSMEnc(x_1,R),\ldots,\allowbreak
\PSMEnc(x_k,R))\allowbreak=f(x_1,\ldots,x_k)]=1.$
\item There exists a randomized algorithm (simulator) $\Sim$ such that for 
$X=x_1,\ldots,x_k$ and for $R$ sampled from $\bit^{r}$, it holds that
$$ \left\{\PSMEnc(x_i,R) \right\}_{i \in [k]} \equiv \Sim(f(x_1,\ldots,x_k)).$$
\end{itemize}
The communication complexity of the PSM protocol is the encoding length $t$ and 
the randomness complexity of the protocol is defined to be $|R|=r$.
\end{definition}

\begin{theorem}[Follows from \cite{IshaiK97}]\label{thm:psm}
For every function $f:(\bit^m)^k \to \bit^{\ell}$ that is computable by an 
$s=s(m,k)$-space TM there is an efficient perfectly secure $\PSM$ protocol 
whose communication complexity and randomness complexity are $O(km\ell \cdot 
2^{2s})$.
\end{theorem}

We describe two additional tools that we will use, secret sharing and 
one-time-pad encryption.
\begin{definition}[Secret Sharing]\label{def:secret-sharing}
Let $x \in \bit^n$ be a message. We say $x$ is {\em secret shared} to $k$ 
shares by choosing $k$ random strings $x^1,\ldots,x^k \in \bit^n$ conditioned 
on $x = \bigoplus_{j=1}^{k}x^j$. Each $x^j$ is called a share, and notice that 
the joint distribution of any $k-1$ shares is uniform over $(\bit^n)^{k-1}$.
\end{definition}

\begin{definition}[One-Time-Pad Encryption]\label{def:one-time-pad}
Let $x \in \bit^n$ be a message. A one-time pad is an extremely simple 
encryption scheme that has information theoretic security. For a random key $K 
\in \bit^n$ the ``encryption'' of $x$ according to $K$ is $\widehat{x}=x \oplus 
K$. It is easy to see 
that the encrypted message $\widehat{x}$ (without the key) is distributed as a 
uniform random string. To decrypt $\widehat{x}$ using the key $K$ we simply 
compute $x=\widehat{x} \oplus K$. The key $K$ might be referenced as the 
encryption key or decryption key.
\end{definition}
\section{Private Neighborhood Trees}\label{subsec:privaencover}
The graph theoretic basis for our compiler is given by \emph{Private 
Neighborhood Trees}, a decomposition of the graph $G$ into (possibly overlapping) trees 
$T(u_1),\ldots,T(u_n)$ such that each tree $T(u_i)$ contains the 
neighbors of $u_i$ in $G$ but \emph{does not} contain $u_i$. Each tree 
$T(u_i)$ provides the neighbors of $u_i$ a way to communicate privately 
without their root $u_i$. The goal is to compute a collection of 
trees with small overlap and small diameter. Thus, we are 
interested in the existence of a \emph{low-congestion} private neighborhood 
trees. 

\begin{definition}[Private Neighborhood Trees]
Let $G=(V=\{u_1,\ldots, u_n\},E)$ be a $2$-vertex connected graph.
The private neighborhood trees $\cN$ of $G$ is a collection of $n$ subtrees 
$T(u_1), \ldots, T(u_n)$ in $G$ such that for every $i \in \{1,\ldots, n\}$ it 
holds that $\Gamma(u_i)\setminus \{u_i\}\subseteq T(u_i)$, but $u_i \notin 
T(u_i)$. 
A $(\dilation,\congestion)$-private neighborhood trees $\cN$ satisfies:
\begin{enumerate}
\item
$\Diam(T(u_i)) \le \dilation$ for every $i \in \{1,\ldots, n\}$,
\item
Every edge $e \in E$ appears in at most $\congestion$ trees.
\end{enumerate}
\end{definition}
Note that since the graph is $2$-vertex connected, all the neighbors of $u$ are indeed connected in $G\setminus\{u\}$ for every node $u$. The main 
challenge is in showing that all $n$ trees can be both of small diameter and with small overlap.

\begingroup
\def\thetheorem{\ref{thm:neighborhood-cover}}
\begin{theorem}[Private Trees]
For every $2$-vertex connected graph $G$ with maximum degree 
$\Delta$ and diameter $D$, there exists a $(\dilation,\congestion)$-private 
trees with $\dilation=O(D\cdot \Delta \cdot\log n)$ and 
$\congestion=O(D \cdot\log \Delta\cdot \log^3 n)$. 
\end{theorem}
\addtocounter{theorem}{-1}
\endgroup
\begin{proof}
The construction of the private neighborhood trees $\cN$ consists of 
$\ell=\log \Delta$ phases. In each phase, we compute an $(O(D\log n),O(\log^3 n))$ 
cycle cover in some auxiliary graph using the cycle cover algorithm in 
\cite{ParterY18}. We begin
by having each node $u$ holding an empty forest $F_0(u)=(\Gamma(u,G),\emptyset)$ 
consisting only of $u$'s neighbors. Then, in each phase we add edges to these 
forests such that the number of connected components (containing the neighbors 
$\Gamma(u,G)$) is reduced by a factor of $2$. After $\ell$ phases, we have a single connected component, that is, we 
have that every $u \in V$ has a tree $T(u)$ in $G\setminus \{u\}$ that spans all neighbors $\Gamma(u,G)$.
Let $\cC_0$ be a cycle cover of $G$. For every phase $t \in \{0,\ldots, 
\ell\}$, let $CC_t(u)$ be the number of connected components in the forest 
$F_t(u)$. Note that $CC_0(u)=\deg(u)$ for all nodes $u$.

In each phase $t \geq 1$, we have a collection of forests 
$\mathcal{N}_{t-1}=\{F_{t-1}(u_1), \ldots, F_{t-1}(u_n)\}$ that satisfy the 
following for every $u_i$:
\begin{enumerate}
	\item $F_{t-1}(u_i) \subseteq G\setminus \{u_i\}$ (the forest avoids $u_i$).
	\item $\Gamma(u_i) \subseteq V(F_{t-1}(u_i))$ (the forest contains all 
	neighbors of $u_i$).
	\item $F_{t-1}(u_i)$ has $CC_{t-1}(u_i)\leq \deg(u_i)/2^{t-1}$ connected 
	components.
\end{enumerate}
It is easy to see that these conditions are met for phase $t=0$, when we have 
the empty forest that simply contains all the neighbors of $u_i$.
The goal of phase $t$ is to add edges to each $F_{t-1}(u_i)$ in order to reduce 
the number of connected components by factor $2$. The algorithm uses the 
current collection of forests $\mathcal{N}_{t-1}$ to define an auxiliary graph 
$\widetilde{G}_t$ which contains the nodes of $G$ and some additional 
``virtual'' nodes and a different set of edges.

For every $u \in V$, we add to $\widetilde{G}_t$ a set of 
$k=CC_{t-1}(u)$ virtual nodes $\widetilde{u}_1,\ldots, \widetilde{u}_k$, one 
for each connected component. We connect $u$ to each of its virtual copies 
$\widetilde{u}_i$ by adding the edges $(u,\widetilde{u}_i)$. Let $(u,v) \in E$ 
be an edge such that $v$ is in the 
$\ith{j}$ connected component of $u$, and $u$ is in the $\ith{i}$ connected 
component of $v$. Then we add the edge $(\widetilde{u_j},\widetilde{v_i})$ to 
the graph $\widetilde{G}_t$.
The graph $\widetilde{G}_t$ has $O(m)$ nodes, $O(m)$ edges and diameter at most $3D$. To see this, any edge $(u,v)$ in the original graph $G$ can be replaced with the path $u \to \widetilde{u}_j \to \widetilde{v}_i \to v$. 

Next, we compute an $(O(D\log n),O(\log^3 n))$-cycle cover $\widetilde{\cC}_t$ 
for the edges of $\widetilde{G}_t$. In the analysis part (Cl. \ref{cl:twovertex}), we show that $\widetilde{G}_t$ is also $2$-vertex connected, and hence all its edges are covered by cycles of $\widetilde{\cC}_t$. 
To map these virtual cycles to real cycles $\cC_t$ in $G$, we simply replace a 
virtual node $\widetilde{u}_j$ with the real node $u$. An edge 
$(u,\widetilde{u}_j)$ will be contracted to just $u$, and edge 
$(\widetilde{u}_j,\widetilde{v}_i)$ will be replaced by $(u,v)$.

Let $G_t(u)$ be the forest $F_{t-1}(u)$ obtained by adding to it all the edges 
of the cycles in $\cC_t$ that intersect $u$, but avoiding the node $u$. That 
is, we define
$$G_t(u) = F_{t-1}(u) \cup \{C \in \cC_t ~\mid~ u\in C\} 
\setminus\{u\}.
$$
Moreover, let $F_t(u) \subseteq G_t(u)$ be a forest that spans all 
the neighbors of $u$. This forest can be computed, for instance, by running a 
BFS from a neighbor of $u$ in each connected component of $G_t(u)$. 
This completes the description of phase $t$. The final private tree collection 
is given by $\cN=\{F_{\ell}(u_1), \ldots, F_{\ell}(u_n)\}$. 
We now turn to analyze this construction and prove 
\Cref{thm:neighborhood-cover}.

\paragraph{Correctness.}
We start by proving each $F_{\ell}(u_i)$ is a tree.
Note that we run the cycle cover algorithm on the graph $\widetilde{G}_t$. 
To obtain a cycle for each edge in $\widetilde{G}_t$ it is required to show that 
$\widetilde{G}_t$ is 2-edge 
connected. We show the stronger property that it is in fact a 2-vertex connected.
\begin{claim}\label{cl:twovertex}
For every phase $t \in [\ell]$, the graph $\widetilde{G}_t$ is 2-vertex 
connected.
\end{claim}
\begin{proof}
Recall that the graph $G$ is 2-vertex connected, and that any edge $(u,v)$ in 
the original graph $G$ can be replaced with the path $u \to \widetilde{u}_j \to \widetilde{v}_i \to v$ in $\widetilde{G}_t$. 

Let $u \in G$ be a (real) node and consider removing a node $u' \in 
\{u,\widetilde{u}_1,\ldots,\widetilde{u}_{k}\}$ from $\widetilde{G}_t$. We will 
show that the graph $\widetilde{G}_t \setminus \{u'\}$ 
remains connected. Let $w_1,w_2$ be two nodes 
in $\widetilde{G}_t \setminus \{u'\}$. 
Let $w'_1$ be the closest node to $w_1$ in $\widetilde{G}_t \setminus \{u'\}$ 
such that $w'_1 \in G \setminus{u}$ (it might be that $w_1=w'_1$).
We observe that such a node $w'_1$ 
always exists: if $w_1 \in G \setminus \{u\}$ then $w'_1=w_1$, if 
$w_1=\widetilde{v}_i$ for some $i$ and $v \ne u$ then $w'_1=v$ and if 
$w_1=\widetilde{u}_i$ for some $i$ then there exists $v \ne u$ such that $(v,u) 
\in G$ and therefore there exists $i'$ such that 
$\widetilde{v}_{i'}$ is connected to $\widetilde{u}_i$ and then $w'_1=v$.
Similarly, 
let $w'_2$ be the closest node to $w_2$ in $\widetilde{G}_t \setminus \{u'\}$ 
such that $w'_2 \in G \setminus \{u\}$ (it might be that $w'_1=w'_2$).

Since 
$G$ is 2-vertex connected, then there is a path $P$ in $G$ between $w'_1$ and 
$w'_2$ that 
does not use $u$. Therefore, there is a path $\widetilde{P}$ between $w'_1$ and 
$w'_2$ in $\widetilde{G}_t$ that does not use $u$ or $\widetilde{u}_j$ for any $j \in [k]$. 
Thus, there is a path in $\widetilde{G}_t$ from $w_1$ to $w_2$ as follows: 
$w_1,\ldots,w'_1,\widetilde{P},w'_2,\ldots,w_2$. Note that it might be the case 
that this path is 
not simple (\eg if $\widetilde{P}$ uses $w_1$) but it can be made simple by omitting internal cycles.
\end{proof}

Next, we show that indeed the final forests are trees.
\begin{claim}\label{lemma:private-n-covering}
For every phase $t \in [\ell]$ and for every $u \in V$ the number 
of connected components satisfies $CC_t(u) \leq \Delta/2^{t}$.
\end{claim}
\begin{proof}
The lemma is shown by induction on $t$. The case of $t=0$ holds vacuously. 
Assume that the claim holds up to $t-1$ and consider phase $t$. 
By construction, for each $u$, the auxiliary graph $\widetilde{G}_t$ contains
$CC_{t-1}(u)$ virtual nodes $\widetilde{u}_j$ that are connected to $u$.

The cycle cover $\widetilde{\cC}_t$ for $\widetilde{G}_t$ covers all these 
virtual edges $(u,\widetilde{u}_j)$ by virtual cycles, each such cycle connects 
two virtual nodes. Since every two virtual nodes of $u$ in $\widetilde{G}_t$ 
are connected to neighbors of $u$ that belong to different components in 
$G_{t-1}(u)$, every cycle that connects two virtual neighbors is mapped into a 
cycle that connects two of $u$'s neighbors that belong to a different connected 
component in phase $G_{t-1}(u)$. Hence, the number of connected components in 
the forest $F_t(u)$ has been decreased by factor at least $2$ compared to that 
of $F_{t-1}(u)$. 
\end{proof}

\paragraph{Small Diameter.}
We show that the diameter of each tree $T(u_i)$ is bounded by $O(\Delta D\cdot 
\log n)$. Note that this bound is existentially tight (up to logarithmic 
factors) as there are graphs $G$ with diameter $D$ and a node $u$ with degree 
$\Delta$ such that the diameter of $G\setminus \{u\}$ is $O(\Delta D)$.  

\begin{claim}
The diameter of each tree $T(u_i) \in \cN$ is $O(\Delta \cdot D\cdot \log n)$.  
\end{claim}
\begin{proof}
Note that the forest $F_t(u)$ is formed by a collection of $O(D\log 
n)$-length cycles that connect $u$'s neighbors. Hence, when removing $u$, we 
get paths of length $O(D\log n)$. Consider the process where in each phase $t$, 
every two $u$'s-neighbors that are connected by a cycle in $\cC_t$ are 
connected by a single ``edge''. By the Proof of 
\Cref{lemma:private-n-covering}, after $\log \Delta$ phases, we get a 
connected tree with $\deg(u)$ nodes, and hence of ``diameter" $\deg(u)$. 
Since each edge corresponds to a path of length $O(D\log n)$ in $G$, we get 
that the final diameter of $F_{\ell}(u)$ is $O(\deg(u)\cdot D \cdot \log n)$. 
\end{proof}

\paragraph{Congestion.} We analyze the congestion of the construction.
\begin{claim}\label{cl:pneighbor-congestion}
Each edge $e$ appears on at most $O(D\log^3 n \log \Delta)$ different trees $T(u_i)\in \cN$.
\end{claim}
\begin{proof}
We first show that the cycles $\cC_t$ computed in $G$ have congestion $O(\log^3 
n)$ for every $t \in \{1,\ldots,\ell\}$. Clearly, the cycles 
$\widetilde{\cC}_t$ computed in $\widetilde{G}_t$ have congestion of $O(\log^3 
n)$. Consider the mapping of cycles  $\widetilde{\cC}_t$ in $\widetilde{G}_t$ 
to a cycles $\cC_t$ in $G$. 
Edges of the type $(u,\widetilde{u}_j)$ are replaced by $(u,u)$ and hence there is no real edge in the cycle.
Edges of the type $(\widetilde{u}_j,\widetilde{v}_i)$ are replaced by $(u,v)$. Since there is 
only one virtual node of $u$ that connects to $v$, and since  
$(\widetilde{u}_j,\widetilde{v}_i)$ appears in $O(\log^3 n)$ many cycles, also $(u,v)$ 
appears in $O(\log^3 n)$ many cycles (i.e., this conversion does
not increase the congestion).

Note that the cycle $C$ of each edge $(u,v)$ joins the $G_t$ subgraphs of at most $D$ nodes since in our construction a cycle $C$ might cover up to $D$ edges. In addition, each edge $e'$ appears on different cycles in $\cC_t$.

We now claim that each edge $e$ appears on $O(t \log^3 n \cdot D)$ graphs $G_t(u)$. 
For $t=1$, this holds as the cycle $C$ of an edge $(u,v)$ joins the subgraphs 
$G_1(x)$ and $G_1(y)$ for every edge $(x,y)$ that is covered by $C$. Assume it 
holds up to $t-1$ and consider phase $t$. In phase $t$, we add to the $G_t(u)$ 
graphs the edges of $\cC_t$. 
Again, each cycle $C'$ of an edge $(u,v)$ joins $D$ graphs $G_t(x), G_t(y)$ for every $(x,y)$ that is covered by $C'$. Hence each edge $e$ appears on $O(D\cdot\log^3 n)$ of the subgraphs 
$G_t(u_j)\setminus G_{t-1}(u_j)$. By induction assumption, each $e$ appears on 
$(t-1)\log^3 n \cdot D$ graphs $G_{t-1}(u_j)$ and hence overall each edge $e$ appears 
on $O(t\log^3 n)$ graphs $G_{t}(u_j)$.
Therefore we get that each edge appears on $O(\log \Delta \cdot \log^3 n \cdot D)$ 
trees in $\cN$.
\end{proof}
The above proof actually shows a slightly stronger statement: a construction of $(\dilation,\congestion)$-cycle covers yields a $(\dilation', \congestion')$ private neighborhood trees for $\dilation'=O(\dilation \cdot \Delta)$ and $\congestion'=O(\congestion \cdot \dilation \cdot \log \Delta)$.
\end{proof} 
The distributed construction of private neighborhood trees is in \Cref{sec:preproc}. 
\section{Secure Simulation via Private Neighborhood Trees}\label{sec:simulation}
In this section, we describe how to transform any distributed algorithm $\cA$ 
to 
a new algorithm $\cA'$ which has the same functionality as $\cA$ (\ie the 
output 
for every node $u$ in $\cA$ is the same as in $\cA'$) but has perfect 
privacy (as is defined in \Cref{def:perfect-privacy}). Towards this end, we 
assume that the combinatorial structures required 
are already computed (in a preprocessing stage described in \Cref{sec:preproc}), namely, a private neighborhood 
tree in the graph. The output of the preprocessing stage is given in a 
distributed manner. The (distributed) output of the 
private neighborhood trees for each node $u$, is such that each vertex $v$ 
knows its parent in the private neighborhood tree of $u$ (if such exists). 

\begingroup
\def\thetheorem{\ref{thm:secure-algorithm}}
\begin{theorem}
Let $G$ be an $n$-vertex graph with diameter $D$ and maximal degree $\Delta$. 
Let $\A$ be a natural distributed algorithm that works on $G$ in $r$ rounds. 
Then, 
$\A$ can be transformed to a new algorithm $\A'$ with the same output 
distribution and which has perfect 
privacy and runs in $\widetilde{O}(rD \cdot \poly(\Delta))$ rounds (after a 
preprocessing 
stage).
\end{theorem}
\addtocounter{theorem}{-1}
\endgroup

As a preparation for our secure simulation, we provide the following convenient 
view of distributed algorithms.  

\subsection{Our Framework}
We treat the distributed 
$r$-round algorithm $\cA$ from the viewpoint of some fixed node $u$, as a 
collection of $r$ functions $f_1,\ldots,f_r$ as 
follows. Let
$\Gamma(u)=\{v_1,\ldots,v_{k}\}$. At any round $i$, the memory of $u$ consists 
of a 
state, denoted by $\sigma_i$ and $\Delta$ messages $m_{v_1 \to u}\ldots, 
m_{v_\Delta \to u}$ that were received in the previous round (in the degree 
of 
the node is less than $\Delta$ the rest of the messages are empty). Initially, 
we set 
$\sigma_0$ to be a fixed string and initialize all messages to NULL. At round 
$i$ the node $u$ updates its state to $\sigma_{i+1}$ according to its previous 
state $\sigma_i$ and the messages that it got in the previous round. It then prepares $k$ 
messages to send $m_{u \to v_1}, \ldots, m_{u \to v_{\Delta}}$. To ease 
notation (and without loss of generality), we assume that each state contains 
the ID of the node $u$. Thus, we can focus on a single update 
function $f_i$ for every round that works for all nodes. The function $f_i$ 
gets the state $\sigma_i$, the messages $m_{v_1 \to u}\ldots, 
m_{v_\Delta \to u}$, and the randomness $s$. The output of $f_i$ is the next 
state $\sigma_{i+1}$, and at most $k$ outgoing 
messages:
\begin{align*}
	(\sigma_{i},m_{u \to v_1}, \ldots, m_{u \to v_\Delta}) \gets 
	f_i(\sigma_{i-1},m_{v_1 \to u},\ldots, m_{v_\Delta \to u}, s).
\end{align*}

Our compiler works \emph{round-by-round} where each round $i$ is replaced by a 
collection of rounds that ``securely'' compute $f_i$, in a manner that will be 
explained next.
The complexity of our algorithm depends exponentially on the space 
complexity of the functions $f_i$. Thus, we proceed by transforming the 
original algorithm $\cA$ to one in which each $f_i$ can be computed in 
logarithmic space, while slightly increasing the number of rounds.

\begin{claim}\label{clm:log-space}
Any natural distributed algorithm $\cA$ the runs in $r$ rounds can be 
transformed 
into a new algorithm $\widehat{\cA}$ with the same output distribution such 
that 
$\widehat{\A}$ is computable in logarithmic space using $r'=r \cdot 
\poly(\Delta + \log n)$ rounds.
\end{claim}
\begin{proof}
Let $t$ be the running time of the function 
$f_i$. Then, $f_i$ can be computed with a circuit of at most $t$ gates. Note that since $\cA$ is natural, it holds that $t\leq \poly(\Delta, \log n)$. 

Instead of letting $u$ computing $f_i$ in round $i$, we replace the $i$th round by $t$ rounds 
where each round computes only a single gate of the function $f_i$. These new rounds 
will have no communication at all, but are used merely for computing $f_i$ with 
a \emph{small} amount of memory.

Let $g_1,\ldots,g_t$ 
be the gates of the function $f_i$ in a computable order where $g_t$ is the 
output of the function. We define a new 
state $\sigma'_i$ of the form $\sigma' = 
(\sigma_i,g_1,\ldots,g_t)$, where $\sigma_i$ is the original state, and $g_j$ 
is the value of the $j$th gate. Initially, $g_1,\ldots,g_t$ are set to $\bot$. 
Then, for all $j \in [t]$ we define the function
$$f^j_i(\sigma_i,g_1,\ldots,g_{j-1},\bot,\ldots,\bot) = 
(\sigma_i,g_1,\ldots,g_{j-1},g_j,\bot,\ldots,\bot).
$$
In the $j$th round, we compute $f^j_i$, until the final $g_t$ is computed. Note 
that $f^j_i$ can be computed with logarithmic space, and since $t \le 
\poly(\Delta, \log n)$ we can compute $f^j_i$ with space $O(\log \Delta + \log 
\log n)$.
As a result, the $r$-round algorithm $\cA$ is replaced by an $rt$-round algorithm $\widehat{\cA}$,
where $t \le \poly(\Delta, \log n)$. That 
is, we have that $r' \le \poly(\Delta, \log n)$.
\end{proof}

As we will see, our compiler will have an overhead of $\poly(\Delta, \log n)$ 
in the round complexity and hence the overhead of \Cref{clm:log-space} is 
insignificant. Thus, we will assume that the distributed algorithm $\cA$ 
satisfies that all its functions $f_i$ are computable in logarithmic space 
(\ie we assume that the algorithm is already after the above transformation). 

\subsection{Secure Simulation of a Single Round}
In the algorithm $\cA$ each node $u$ computes the function $f_i$ in each round 
$i$. In our secure algorithm $\cA'$ we want to simulate this computation, however, 
on \emph{encrypted} data, such that $u$ does not get to learn the true output of $f_i$ 
in any of the rounds except for the last one. When we say ``encrypted'' data, we 
mean a ``one-time-pad'' (see \Cref{def:one-time-pad}). That is, we merely refer 
to a process where we the data is masked by XORing it with a random string $R$. 
Then, $R$ is called the encryption (and also decryption) key. Using this 
notion, we define a related function 
$f'_i$ that, intuitively, simulates $f_i$ on encrypted data, by getting 
encrypted state and messages as input, decrypting them, then computing $f_i$ 
and finally encrypting the output with a new key. We simulate every round of 
the 
original algorithm $\A$ by a $\PSM$ protocol for the function $f'_i$.

\paragraph{The Secure Function $f'_i$.}
The function $f'_i$ gets the following inputs (encrypted elements will 
be denoted by the $\widehat{\cdot}$ notation):
\begin{enumerate}
	\item An encrypted state $\widehat{\sigma}_{i-1}$ and encrypted messages 
	$\curparenn{\widehat{m}_{v_j \to u}}_{j=1}^{\Delta}$.
	\item The decryption key $R_{\sigma_{i-1}}$ of the state $\widehat{\sigma}_{i-1}$
	 and the decryption keys $\curparenn{R_{v_j \to u}}_{j}^{\Delta}$
	for the messages $\curparenn{\widehat{m}_{v_j \to u}}_{j=1}^{\Delta}$.	
	\item Shares for randomness $\curparenn{R^j_s}_{j=1}^{\Delta}$ for the 
	function $f_i$.
	\item Encryption keys for encrypting the new state
	$R_{\sigma_{i}}$ and messages
		$\curparenn{R_{u \to v_j}}_{j=1}^{\Delta}$.
\end{enumerate}
The function $f'_i$ 
decrypts the state and messages and 
runs the function $f_i$ (using randomness $s=\bigoplus R^j_s$) to get the new 
state 
$\sigma_{i}$ and the outgoing messages $m_{u \to v_1}, \ldots, m_{u \to v_\Delta}$. 
Then, it encrypts the new state and messages using the encryption keys.
In total, the function $f'_i$ has $O(\Delta)$ input bits.
The precise description of $f'_i$ is given in \Cref{fig:func-f'}.
\begin{figure}[!h]
\begin{boxedminipage}{\textwidth}
\vspace{1mm} \textbf{The description of the function $f'_i$.}
\\ \textbf{Input:}
An encrypted state $\widehat{\sigma}_{i-1}$,
encrypted messages $\curparen{\widehat{m}_{v_j \to u}}_{j=1}^{\Delta}$, 
keys for decrypting the input 
$R_{\sigma_{i-1}}, \curparen{R_{v_j \to u}}_{j}^{\Delta}$,
randomness $\curparen{R^j_s}_{j=1}^{\Delta}$
and keys for encrypting the output
$R_{\sigma_{i}}, \curparen{R_{u \to v_j}}_{j=1}^{\Delta}$.

\vspace{1mm} \textbf{Run:}
\begin{enumerate}
	\item Compute $\sigma_{i-1} \gets 
	\widehat{\sigma}_{i-1} \oplus R_{\sigma_{i-1}}$ and $s 
	\gets \paren{\bigoplus_{j=1}^{\Delta} R^j_s}$.
	\item For $j=1 \ldots \Delta$: compute $m_{v_j \to u} \gets 
	\widehat{m}_{v_j \to u} \oplus R_{v_j \to 
	u}$.
	\item Run $	\sigma_{i},m_{u \to v_1}, \ldots, m_{u \to v_\Delta} \gets 
		f(\sigma_{i-1},m_{v_1 \to u},\ldots, m_{v_\Delta \to u}, s)$.
	\item Compute $\widehat{\sigma_{i}} \gets \sigma_{i} \oplus 
	R_{\sigma_{i}}$.
	\item For $j=1 \ldots \Delta$: compute $\widehat{m}_{u \to v_j} \gets 
	m_{u \to v_j} \oplus R_{u \to v_j}$.
	\item Output $\widehat{\sigma_{i}},\widehat{m}_{u \to 
	v_1},\ldots,\widehat{m}_{u \to v_\Delta}$.
\end{enumerate}
\end{boxedminipage}

\caption{The function $f'_i$.}
\label{fig:func-f'}
\end{figure}

Recall that in the $\PSM$ model, we have $k$ parties $p_1,\ldots, p_k$ and a 
server $s$, 
where it was assumed that all parties have private shared randomness (not 
known to $s$). Our goal is to compute $f'_i$ securely by implementing a $\PSM$ 
protocol for all nodes in the graph simultaneously.

The compiler securely 
computes $f'_i$ by simulating the $\PSM$ 
protocol for $f'_i$, treating $u$ as the \emph{server} and its immediate 
neighborhood as the \emph{parties}. In order to exchange the private 
randomness, we use the notion of private neighborhood trees.

A private neighborhood tree collection consists of $n$ trees, a tree 
$T_u$ for every $u$, 
that spans all the neighbors of $u$  (\ie the parties) without going through $u$, i.e.,
$T_u \subseteq G\setminus \{u\}$. Using this 
tree, all the parties can compute shared private random bits $R$ which  are not known to $u$. 
For a single node $u$, this can be done in $O(\Diam(T_u)+|R|)$ rounds, where $\Diam(T_u)$ is the 
diameter of the tree and $R$ is the number of random bits. Clearly, our objective is to have trees $T_u$ with small diameter.  Furthermore, as we wish to implement this kind of communication in all $n$ trees, $T_{u_1},\ldots, T_{u_n}$ simultaneously, a second objective is to have small overlap between the trees. That is, we would like each edge $e$ to appear only on a small number of trees $T_u$ (as on each of these trees, the edge is required to pass through different random bits). These two objectives are encapsulated in our notion of 
\emph{private-neighborhood-trees}. 
The final algorithm $\A'_i(u)$ 
for securely computing $f'_i$ is described in \Cref{fig:alg-psm}.

\begin{figure}[!h]
\begin{boxedminipage}{\textwidth}
\vspace{1mm} \textbf{The algorithm $\A'_i(u)$ for securely computing $f'_i$.}

\textbf{Input:} Each node $v \in \Gamma(u)$ has input $x_v \in \bit^{m}$.
\begin{enumerate}
	\item Let $T_u$ be the tree spanning $\Gamma(u)$ in $G \setminus \{u\}$ and 
	let $w$ be the root.
	\item $w$ chooses a random string $R$ and sends it to $\Gamma(u)$ 
	using the tree $T_u$.
	\item Each node $v \in \Gamma(u)$ computes $M_v = 
	\PSMEnc(f'_i,x_v,R)$ and sends it to $u$.
	\item $u$ computes $y=\PSMDec \left(f'_i,\{M_v\}_{v \in \Gamma(u)} 
	\right)$.
\end{enumerate}
\end{boxedminipage}
\caption{The description of the distributed $\PSM$ algorithm of node $u$ for 
securely computing the function $f'_i$.}
\label{fig:alg-psm}
\end{figure}

In what follows analyze the security and round complexity of Algorithm $\A'_i$.

\paragraph{Round Complexity.}
Let $f:\{0,1\}^{m \cdot|\Gamma(u)|} \to \{0,1\}^\ell$ be a function with 
$|\Gamma(u)| \le \Delta$ inputs, where each input is of length $m$ 
bits. The communication complexity 
of the $\PSM$ protocol depends on the input and output length of the function 
and also on the memory required to compute $f$. Suppose that $f$ is computable 
by an $s$-space Turing Machine (TM). Then, by \Cref{thm:psm} the communication 
complexity (and randomness complexity) of the protocol is at most 
$O(\Delta m\ell \cdot 2^{2s})$.

In the first phase of the protocol, the root $w$ sends a collection of random 
bits $R$ to $\Gamma(u)$ 
using the private neighborhood trees, where $|R|=O(\Delta \cdot m\cdot\ell 
\cdot 
2^{2s})$. By 
\Cref{thm:neighborhood-cover}, the diameter 
of the tree is at most $\widetilde{O}(D\Delta)$ and each edge belongs to 
$\widetilde{O}(D)$ different trees. Therefore, there are total of 
$\widetilde{O}(D \cdot |R|)$ many bits that need to go through a single 
edge when sending the information on all trees simultaneously.
Using the random delay approach of
\Cref{thm:delay}, this can be done in 
$\widetilde{O}(D\Delta+D \cdot|R|)=\widetilde{O}(\Delta \cdot D\cdot m 
\cdot\ell \cdot 
2^{2s})$ rounds.
This is summarized by the following Lemma:

\begin{lemma}\label{lemma:distributed-psm}
	Let $f:(\bit^m)^{\Delta} \to \bit^{\ell}$ be a function over $\Delta$ 
	inputs where each is of length at most $m$ and that is computable by a 
	$s$-space Turing Machine. Then, there is a distributed algorithm $\A'_i(u)$ (in the 
	$\mathsf{CONGEST}$ model) with perfect 
	privacy where each node $u$ outputs $f$ evaluated on $\Gamma(u)$. The 
	round complexity of $\A'_i(u)$ is $\widetilde{O}(\Delta \cdot D\cdot m 
\cdot\ell \cdot 
2^{2s})$.
\end{lemma}

\subsection{The Final Secure Algorithm}
Using the function $f'_i$, we define the algorithm $\cA'_u$ for computing the 
next state and messages of the node $u$. We describe the algorithm for any $u$ in 
the graph, and at the end, we show that all the algorithms $\{\cA'_u\}_{u \in 
G}$ 
can be run simultaneously with low congestion.

The algorithm $\cA'_u$ involves running the distributed algorithm $\A'_i(u)$ 
for each round $i \in \{1,\ldots, r\}$. The secure simulation of round $i$ starts by letting the root of each tree $T_u$ (i.e., the tree connecting the neighbors of $u$ 
in $G \setminus \{u\}$) sample a key $R_{\sigma_{i}}$ for encrypting the new 
state of $u$. Moreover, each neighbor $v_j$ of $u$ samples 
a share of the randomness $R^j_s$ used to evaluate the function $f_i$, and a 
key $R_{u \to v_j}$ for encrypting the message sent from $u$ to $v_j$.

Then they run $\A'_i(u)$ algorithm with $u$ as the server and 
$\Gamma(u)$ as the parties for computing the function $f'_i$ (see 
\Cref{fig:alg-psm}). The node $u$ has the 
encrypted state and message, the neighbors of $u$ have the (encryption and 
decryption) keys for the current state, the next state and the sent messages, and 
moreover the randomness for evaluating $f'_i$. At the end of the 
protocol, $u$ computes the output 
of $f'_i$ which is the encrypted output of the function $f_i$.

After the final round, 
$u$ holds an encryption of the final state $\widehat{\sigma_r}$ which contains 
only the output of the original algorithm $\A$. At this point, the neighbors of 
$u$ send it 
the decryption key for this last state, $u$ decrypts its state and outputs the 
decrypted state. Initially, the state $\sigma_{0}$ is a fixed string which is 
not encrypted, and the encryption keys for this round are assumed to be 
$0$. The description is summarized in \Cref{fig:alg-A_u}. See \Cref{fig:secrets} for an illustration.
\begin{figure}[!h]
\begin{boxedminipage}{\textwidth}
\vspace{1mm} \textbf{The description of the algorithm $\cA'_u$.}
\begin{enumerate}
  \item Let $v_1,v_2, \ldots, v_\Delta$ be some arbitrary ordering on $\Gamma(u)$. 
	\item For each round $i=1\ldots r$ do:
	\begin{enumerate}
	\item $u$ sends $\widehat{\sigma}_{i-1}$ to neighbor $v_2$.
	\item Each neighbor $v_j$ of $u$ samples $R^j_s$ at random (and stores it).
	\item $v_1$ chooses $R_{\sigma_{i}}$ at random (and stores it).
	\item Run the $\A_i(u)$ algorithm for $f'_i$ with server $u$ and 
	parties $\Gamma(u)$ where:
	\begin{enumerate}
		\item $v_1$ has an inputs $R_{\sigma_{i-1}}$ and 
		$R_{\sigma_{i}}$ and $v_2$ has input $\widehat{\sigma}_{i-1}$.
		\item In addition, each neighbor $v_j$ of $u$ has input $R_{u \to 
		v_j},R^j_s$.
		\item $u$ learns the final output of the algorithm
		$(\widehat{\sigma_{i}},\widehat{m}_{u \to 
		v_1},\ldots,\widehat{m}_{u \to v_\Delta})$.	
	\end{enumerate}	
	\end{enumerate}
	\item $v_1$ sends $R_{\sigma_{r}}$ to $u$.
	\item $u$ computes $\sigma_r = \widehat{\sigma}_r \oplus 
	R_{\sigma_{r}}$ and outputs $\sigma_r$.
\end{enumerate}
\end{boxedminipage}

\caption{The description of the Algorithm $\cA'_u$. We 
assume that in ``round 0'' all keys are initialized to 0. That is, 
we let $R_{\sigma_{0}}=0$, and initially set 
$R_{v_j \to u}=0$ for all $j \in [\Delta]$.}
\label{fig:alg-A_u}
\end{figure}

\begin{figure}[h!]
	\begin{center}
		\includegraphics[scale=0.60]{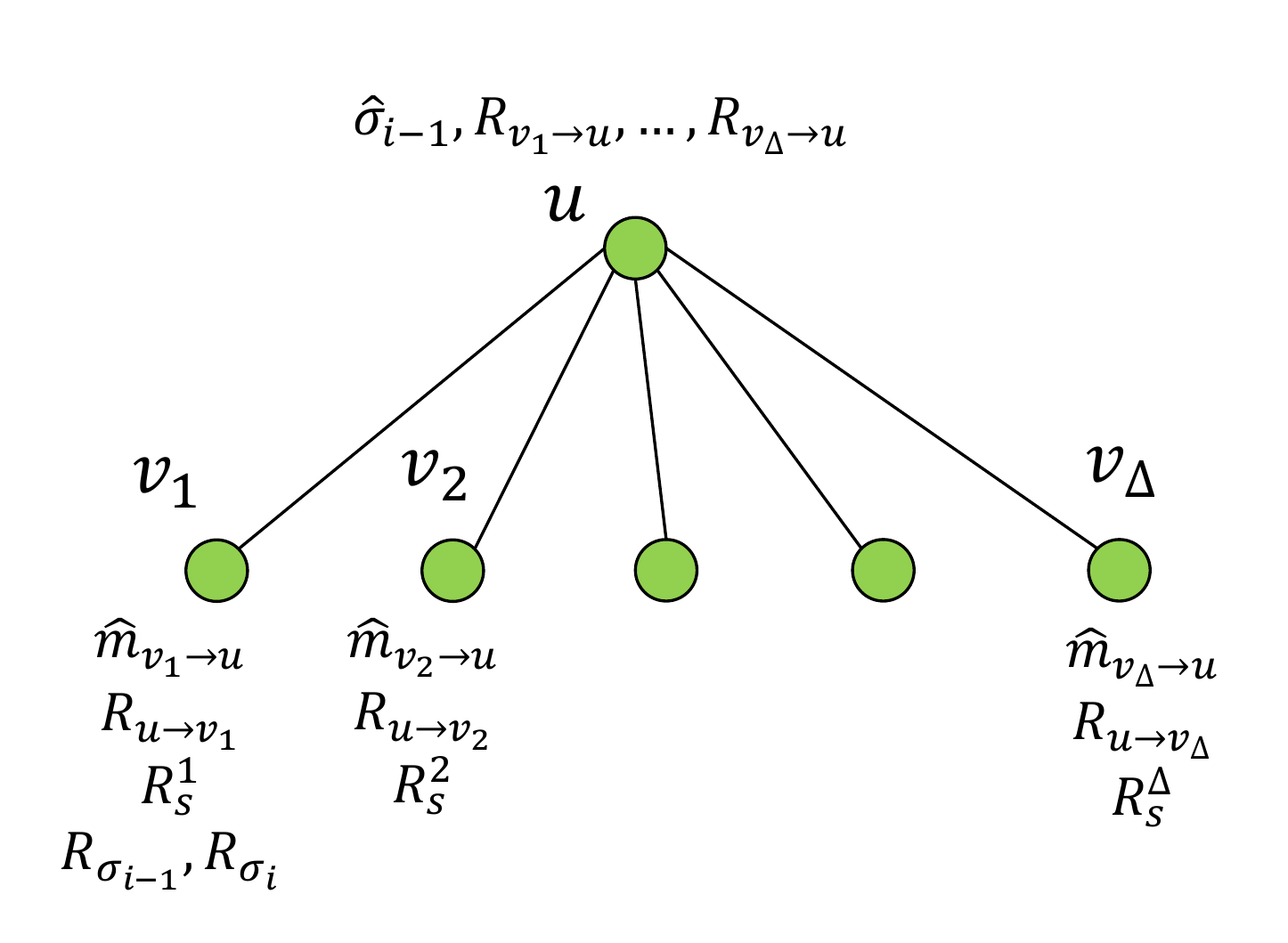}
		\caption{The information held by $u$ and its neighbors in phase $i$ of the algorithm.}
			\label{fig:secrets}
	\end{center}
\end{figure}

Finally, we show that the protocol is correct and secure.

\paragraph{Correctness.}
The correctness follows directly from the construction. Consider a node $u$ in 
the graph. Originally,  $u$ computes the sequence of states 
$\sigma_{0},\ldots,\sigma_r$ where $\sigma_r$ contained the final output of the 
algorithm. In the compiled algorithm $\A'$, for each round $i$ of $\cA$ and 
every node $u$ the 
sub-algorithm $\A'_i(u)$ computes $\widehat{\sigma_{i}}$, where 
$\widehat{\sigma_{i}} = \sigma_{i} 
\oplus R_{\sigma_{i}}$ where $v_1$ holds $R_{\sigma_{i}}$. Thus, after the last 
round, $u$ has 
$\widehat{\sigma_{r}}$ and $v_1$ has $R_{\sigma_{r}}$. Finally, $u$ computes 
$\widehat{\sigma}_r 
\oplus R_{\sigma_{r}}=\sigma_r$ and outputs 
$\sigma_r$ as required.

\paragraph{Round Complexity.}
We compute the number of rounds of the algorithm for any natural algorithm 
$\A$. 
The algorithm consists of $r'=r \cdot \poly(\Delta + \log n)$ iterations. In 
each iteration, every vertex $u$ implements algorithm 
$\A'_i$ for the function $f'_i$ (there are other operations in 
the iteration but they are negligible). We know that $f_i$ can be computed in 
$s$-space where $s=O(\log \Delta + \log\log n)$, and thus we can bound the size 
of each input 
to $f'_i$ by $\poly(\Delta) \cdot \polylog(n)$. Indeed, the state has this 
bound by the definition of a natural algorithm, and thus also the encrypted 
state (which has 
the exact same size), the messages and encryption keys for the messages have 
length at most $\log n$, and the randomness shares are of size at most the 
running time of $f_i$ which is at most $2^s$ where $s$ is the space of $f_i$  
and thus the bound holds. 
The output length shares the same bound as well.

Since $f_i$ can be computed in $s$-space where $s=O(\log \Delta + \log\log n)$, 
we observe that $f'_i$ can be computed in $s$-space as well. This includes 
running $f_i$ in a ``lazy'' manner. That is, whenever the TM for computing 
$f_i$ asks to read a the $\ith{i}$ bit of the input, we generate this bit 
by performing the appropriate XOR operations for the $\ith{i}$ bit of the input 
elements. The memory required for this is only storing indexes of the input 
which is $\log(\Delta \cdot \poly(\log n))$ bits and thus $s$ bits suffice.

Then, by 
\Cref{lemma:distributed-psm} we get that algorithm
$\A'_i(u)$ for $f'_i$ runs in $\widetilde{O}(D \cdot \poly(\Delta))$ rounds, 
and the total number of rounds of our algorithm is  
$\widetilde{O}(rD \cdot \poly(\Delta))$. In particular, if the degree $\Delta$ 
is 
bounded by $\polylog(n)$ then we get $\widetilde{O}(rD)$ number of rounds.
\begin{remark}[Round complexity for non-natural 
algorithms]\label{remark:general}
 If $\cA$ is not a ``natural'' algorithm then we can bound the number of rounds 
 with dependency on the time complexity of the algorithm. If each function 
 $f_i$ (the local computation of the nodes) 
 can be computed by a circuit of size $t$ then the number of rounds of the 
 compiled algorithm is bounded by $\widetilde{O}(rDt \cdot \poly(\Delta))$.
\end{remark}

\paragraph{Security.}
We begin by describing the security of a single sub-protocol $\cA'_u$ for any node 
$u$ in the graph. The algorithm $\cA'_u$ has many nodes involved, and we begin by 
showing how to simulate the messages of $u$. Fix an iteration $i$, and consider 
the all the messages sent to $u$ by the $\PSM$ protocol in $\cA'_i(u)$ denoted 
by 
$\{M_v\}_{v \in G}$, and let $\widehat{\sigma_{i}},\widehat{m}_{u \to 
v_1},\ldots,\widehat{m}_{u \to v_\Delta}$ be the output of the protocol.
By the security of the $\PSM$ protocol, there is a simulator 
$\Sim$ such that the following two distributions are equal:
$$ \curparen{M_v}_{v \in G} \equiv \Sim(\widehat{\sigma_{i}},\widehat{m}_{u \to 
v_1},\ldots,\widehat{m}_{u \to v_\Delta}).$$
Since $\widehat{\sigma_{i}}$ and $\widehat{m}_{u \to 
v_1},\ldots,\widehat{m}_{u \to v_\Delta}$ are encrypted by keys that are 
never sent to $u$ we have that from the viewpoint of $u$ the distribution of 
$\widehat{\sigma_{i}}$ and of $\widehat{m}_{u \to 
v_1},\ldots,\widehat{m}_{u \to v_\Delta}$ are uniformly random. Thus, we can 
run the simulator with a random string $R$ of the same length and have
$$\Sim(\widehat{\sigma_{i}},\widehat{m}_{u \to v_1},\ldots,\widehat{m}_{u \to 
v_\Delta}) \equiv \Sim(R).$$

While this concludes the simulator for $u$, we need to show a simulator for 
other nodes that participate in the protocol. Consider the neighbors of $u$. 
The neighbor $v_1$ has the encryption key for the state, and $v_2$ has the 
encrypted state. Since they never exchange this information, each of them gets 
a uniformly random string. In addition to their own input, the neighbors have 
the shared randomness for the $\PSM$ protocol. All these elements are uniform 
random 
strings which can be simulated by a simulator $\Sim$ by sampling a random 
string of the same length. 

To conclude, the privacy of $\cA'_i(u)$ follows from the perfect privacy of 
$\PSM$ 
protocol we use. The $\PSM$ security guarantees a perfect simulator for 
the server's viewpoint, and it is easy to construct a 
simulator for all other parties in the protocol as they only receive random 
messages. While the $\PSM$ was proven secure in a stand-alone setting, 
in our protocol we have a composition of many instances of the protocol. 
Fortunately, it was shown in \cite{KushilevitzLR10} that any protocol
that is perfectly secure and has a black-box non-rewinding simulator, is also 
secure under universal composability, that is, security is guaranteed to hold 
when many arbitrary 
protocols are performed concurrently with the secure protocol. We observe that 
the $\PSM$ has a simple simulator that is black-box and non-rewinding, and thus 
we can apply the result of \cite{KushilevitzLR10}. This is since the simulator 
of the $\PSM$ protocol is an algorithm that runs the protocol on an arbitrary 
message that agrees with the output of the function.

\appendix
\section{Distributed Construction of Private Neighborhood Trees}\label{sec:preproc}
The distribute output format of private neighborhood trees $\cN$ is that each 
node $u$ knows its parent in the spanning tree $T(v) \in \cN$ for every $v \in 
V$. For the purpose of our compiler, the private neighborhood trees should be 
computed once, in a preprocessing step. We now use the construction of cycle 
covers from \cite{ParterY18}, and show:
\begin{lemma}\label{lem:simulationdist_lcnc}
Given an $r$-round algorithm for constructing $(\dilation, \congestion)$ cycle cover $\mathcal{C}$, there exists an $r'$-round algorithm for construction a $(\dilation', \congestion')$ private neighborhood trees with
$\dilation'=\dilation \cdot \Delta$, $\congestion'=\congestion \cdot \dilation$ and $r'=r \cdot \widetilde{O}(\dilation, \congestion)$. 
\end{lemma}
\begin{proof}
Let $\cA$ be an $r$-round algorithm for computing a $(\dilation, \congestion)$ cycle cover $\cC$.
Using the random delay approach \Cref{thm:delay}, we can make each edge $(u,v)$ know the edges of all the cycles it belongs to in $\cC$ with $\widetilde{O}(\dilation+\congestion)$ rounds. We then mimic the centralized reduction to cycle cover. In this reduction, we have $O(\log \Delta)$ applications of Algorithm $\cA$ on some virtual graph. Since a node $v$ knows the cycles of its edges, it knows which virtual edges it should add in phase $i$. Simulating the virtual graph can be done with no extra congestion in $G$. In each phase $i$, we compute a cycle cover in the virtual graph and then translate it into a cycle cover $\cC_i$  in the graph $G$. By the same argument as in \Cref{cl:pneighbor-congestion}, translating these cycles to cycles in $G$ does not increase the congestion. Using $\widetilde{O}(\dilation+\congestion)$ rounds, each edge $e$ can learn all the edges on the cycles that pass through it appears in $\cC_i$. 
At the last phase $\ell=O(\log \Delta)$, the graph $G_\ell(u_j)$ consists of $O(\log \Delta \cdot \Delta)$ cycles. In particular,
$$G_\ell(u_j)=\bigcup_{i=1}^\ell \{ C \in \cC_i ~\mid~ (u_j,v) \in C, v \in \Gamma(u_j)\}.$$
By the same argument of \Cref{cl:pneighbor-congestion}, each edge $e$ appears on $O(\log \Delta \cdot \congestion\cdot \dilation)$ different subgraphs $G_\ell(u_j)$ for $u_j \in V$. The diameter of each subgraph $G_\ell(u_j)$ can be clearly bounded by the number of nodes it contained which is $O(\log \Delta \cdot \Delta \cdot \dilation)$. Since each edge $e$ knows all cycles it appears on\footnote{We say that an edge $(u,v)$ knows a piece of information, if at least one of the edge endpoints know that.}, it also knows all the graphs $G_\ell(u_j)$ to which it belongs. Computing a spanning tree in $G_{\ell}(u_i)\setminus \{u_i\}$ can be done in $\widetilde{O}(\Delta \cdot \dilation)$ rounds. Using random delay again, and using the fact that each edge appears on $\widetilde{O}(\dilation)$ trees, all the spanning trees in $G_\ell(u_j)\setminus \{u_j\}$ can be constructed simultaneously in $\widetilde{O}(\Delta \cdot \dilation)$ rounds.
\end{proof}
Using the $\widetilde{O}(n)$-round construction of $(\dilation,\congestion)$ cycle covers with $\dilation=\widetilde{O}(D)$ and $\congestion=\widetilde{O}(1)$ from \cite{ParterY18}, yields the following:
\begin{corollary}\label{cor:distprivate}
For every $n$-vertex graph $G=(V,E)$ with diameter $D$ and maximum degree $\Delta$, one can construct in $\widetilde{O}(n+\Delta \cdot \dilation)$ rounds a $(\dilation,\congestion)$ private trees with $\dilation=\widetilde{O}(D \cdot \Delta)$ and $\congestion=\widetilde{O}(D)$.
\end{corollary}
\section*{Acknowledgments}
We thank Benny Applebaum, Uri Feige, Moni Naor and David Peleg for fruitful 
discussions 
concerning the nature of distributed algorithms and secure 
protocols.

\bibliographystyle{alpha}
\bibliography{crypto}
\end{document}